\documentclass[twocolumn,journal]{IEEEtran}
\usepackage[T1]{fontenc}
\usepackage{bm}
\usepackage{amsmath}
\usepackage{amsthm}
\usepackage{amssymb}
\usepackage{graphicx}
\usepackage[unicode=true,
 bookmarks=true,bookmarksnumbered=true,bookmarksopen=true,bookmarksopenlevel=1,
 breaklinks=false,pdfborder={0 0 0},pdfborderstyle={},backref=false,colorlinks=false]
 {hyperref}
\hypersetup{pdftitle={Your Title},
 pdfauthor={Your Name},
 pdfpagelayout=OneColumn, pdfnewwindow=true, pdfstartview=XYZ, plainpages=false}

\makeatletter
\theoremstyle{plain}
\newtheorem{thm}{\protect\theoremname}
\theoremstyle{definition}
\newtheorem{defn}[thm]{\protect\definitionname}
\theoremstyle{remark}
\newtheorem{rem}[thm]{\protect\remarkname}
\theoremstyle{plain}
\newtheorem{lem}[thm]{\protect\lemmaname}
\theoremstyle{plain}
\newtheorem{prop}[thm]{\protect\propositionname}

\usepackage[vlined,boxed,ruled,linesnumbered,resetcount]{algorithm2e}
\usepackage{algpseudocode}
\usepackage{xcolor}
\usepackage[noadjust,sort,compress]{cite}

\@ifundefined{showcaptionsetup}{}{%
 \PassOptionsToPackage{caption=false}{subfig}}
\usepackage{subfig}
\makeatother

\providecommand{\definitionname}{Definition}
\providecommand{\lemmaname}{Lemma}
\providecommand{\propositionname}{Proposition}
\providecommand{\remarkname}{Remark}
\providecommand{\theoremname}{Theorem}

\begin{document}
\title{On Low-complexity Lattice Reduction Algorithms for Large-scale MIMO Detection: the Blessing of Sequential Reduction}
\author{Shanxiang Lyu, Jinming Wen, Jian Weng and Cong Ling
	\thanks{
		{ This
			work was supported in part by the National Natural Science Foundation of China under Grants 61902149,  61932010 and 11871248, in part by the Major Program of Guangdong Basic and Applied Research under Grant 2019B030302008, in part by the Natural Science Foundation of Guangdong Province under Grants 2019B010136003 and 2019B010137005, in part by the Fundamental Research Funds for the Central Universities under Grant 21618329,  and in part by the National Key R\&D Plan of
			China under Grants 2017YFB0802203 and  2018YFB1003701. This paper was presented in part at the 8th International Conference on
			Wireless Communications and Signal Processing (WCSP), Yangzhou, China, Oct. 2016.}
	}
	\thanks{
		S. Lyu, J. Wen and J. Weng are with the College of Cyber Security, Jinan University, Guangzhou 510632, China (e-mail: shanxianglyu@gmail.com,
jinmingwen1@163.com, cryptjweng@gmail.com).
S. Lyu is also with the State Key Laboratory of Cryptology, P.O. Box 5159, Beijing, 100878, China.
} 
	\thanks{C. Ling is with the Department of Electrical
	and Electronic Engineering, Imperial College London, London SW7 2AZ,
	United Kingdom (e-mail: cling@ieee.org). }
}
\maketitle
\begin{abstract}
Lattice reduction is a popular preprocessing strategy in multiple-input
  multiple-output (MIMO) detection.
{\textcolor{black}{In a quest for developing a low-complexity reduction algorithm for large-scale problems, this paper investigates a new framework called sequential reduction (SR),}}
 which aims to reduce the lengths of all basis vectors.
The performance upper bounds of the strongest reduction in SR
are given when the lattice dimension is no larger than $4$. The proposed new framework
enables the implementation of a hash-based low-complexity lattice
reduction algorithm, which becomes especially tempting when applied
to large-scale MIMO detection. Simulation results 
show that, compared to other reduction algorithms,  the hash-based SR algorithm exhibits the lowest complexity while maintaining comparable error performance.
\end{abstract}

\begin{IEEEkeywords}
lattice reduction, MIMO, large-scale,  hash-based.
\end{IEEEkeywords}

\section{Introduction}

The number of antennas has been scaled up to tens or hundreds in multiple-input multiple-output (MIMO) systems to fulfill the performance requirements
needed by the next generation communication systems \cite{Rusek2013}.
A critical challenge that comes with very large arrays is to design
reliable and computationally efficient detectors. Though the well-known
maximum likelihood detector (MLD) provides optimal error performance,
it suffers from exponential complexity that grows with the number of transmit
antennas \cite{Agrell2002}. In the past two decades, lattice-reduction-aided suboptimal detection techniques have been well investigated
\cite{Yao2002,Windpassinger2004,DBLP:journals/tsp/WangHL19}, whose instantaneous complexity
does not depend on constellation size and noise realizations, but
collect the same diversity as the MLD for MIMO systems \cite{Taherzadeh2007b,Taherzadeh2007a,Ma2008}. Although conventional lattice
reduction algorithms suffice for small-scale MIMO systems, there is still
an avenue to pursue a more practical low-complexity reduction algorithm
for large-scale systems. Moreover, an efficient reduction algorithm
for large-scale problems may also find its applications to cryptanalysis
\cite{Nguyen2010} and image processing \cite{Neel2001}.

The principle of designing a reduction algorithm varies depending
on the desired basis properties: to make all the basis vectors short,
or to make the condition number of the reduced basis small. There
are several popular types of lattice reduction strategies, such as
Minkowski reduction, Korkine-Zolotareff reduction (KZ) \cite{Korkinge1877},
Gauss reduction \cite{Micciancio2002}, Lenstra--Lenstra--Lov\'asz
(LLL) reduction \cite{Lenstra1982}, Seysen reduction \cite{Seysen1993},
etc. They yield reduced bases with shorter or more orthogonal basis
vectors, and provide a trade-off between the quality of the reduced
basis and the computational effort required for finding it. In essence,
a reduction algorithm aims to find a unimodular matrix to
transform an input basis into another one with better property. The
process involves a series of elementary operations noted as reflection,
swapping, and translation. These operations vary for distinct algorithms.

Much work has been done to advance conventional reduction algorithms.
Regarding KZ, refs. \cite{Zhang2012tsp,Lyu2017,WenTIT18} give some
practical implementations and improve the performance bounds. As for
blockwise KZ, its faster implementations and the expected basis properties
are given in \cite{Chen2011b}. Researchers have also been
constructing and analyzing the variants of LLL with great effort. For instance, the size reduction step is optimized in \cite{Lyu2017, Ling2010, Zhang2012},
the implementation order of swaps is simplified in \cite{Vetter2009,Wen2014b,Wen2016},
and the fixed complexity versions of LLL are given in \cite{Ling2010,Wen2014b,DBLP:journals/icl/WenC17}. In contrast, the direction on Seysen reduction has few follow-up studies \cite{Seethaler2007,ma2010seysen}, partly because of the fact that Seysen reduction has unsatisfactory performance in high dimensions.

While LLL and blockwise KZ are still the default choices in cryptography to reduce a
basis in hundreds of dimensions, the element-based lattice reduction
($\mathrm{ELR}$) proposed in \cite{Zhou2013} has become more attracting in large
MIMO, which preprocesses a large basis with even lower complexity
than LLL. Later ref. \cite{Zhou2013} has been  generalized to $\mathrm{ELR}^{+}$
  \cite{Zhou2013a} for small-scale problems, but the theoretical characterization of  $\mathrm{ELR}$ and $\mathrm{ELR}^{+}$
has not been given a rigorous treatment, even for small dimensions. {\color{black}It is noteworthy that $\mathrm{ELR}$ and $\mathrm{ELR}^{+}$   have totally different structures with LLL variants, and one 
	might be lured into the belief that $\mathrm{ELR}$ and $\mathrm{ELR}^{+}$  can be tuned to arrive at more sophisticated methods. Nevertheless, no analytical skills can be inherited from LLL/KZ literature \cite{Lenstra1982,Micciancio2002}, which makes the performance analysis of the new algorithms complicated.}


{\color{black}

In
this work, we investigate a general form of  $\mathrm{ELR}$ and $\mathrm{ELR}^{+}$  which we refer to as   
sequential reduction (SR). We derive the objective function  from a MIMO detection task, and present the general form of an SR algorithm which can solve/approximate the smallest basis problem.  Unlike  KZ or Minkowski reduction, SR reduces the
basis vectors by using sub-lattices so as to avoid a basis expansion
process. The strongest algorithm in SR tries to minimize the length of basis vectors
with the aid of a closest vector problem (CVP) oracle.  We show bounds on the
basis lengths and orthogonal defects for small dimensions. 
After that, the feasibility of applying SR to reduce a large dimensional
basis is analyzed, and we actually construct a hash-based algorithm for this task.
Our
simulation results then show the plausibility of using SR
in large-scale MIMO systems.

Preliminary results of this work have been partly presented
in a conference paper \cite{Lyu2016}. Compared with \cite{Lyu2016}, this work contains the following new contributions:

\begin{itemize}
	\item The performance bounds on small dimensional bases are rigorously analyzed (Theorems 2 and 3). Unlike the results in \cite{Lyu2016} that rely on an assumption about covering radius, these bounds hold for all input bases.
	\item Comparisons with other types of strong\&weak reduction are made (Section III-B, Section IV-D), including $\eta$-Greedy reduction, KZ and its variants, Minkowski reduction, LLL and its variants, and Seysen reduction.
	\item A Hash-based SR algorithm is constructed (Section IV). More specifically, the nearest neighbor search problem is approximately solved
	with the aid of hashing, and not through a brute-force search. 
	\item  The theoretical studies are supported with more simulation results (Section V), these include the comparisons with major lattice-reduction-aided MIMO detection algorithms, and the BER performance tested for various channels . 
	\item The types of bases feasible for using SR-Hash is discussed (Appendix A). We numerically show that the dual of large-scale Gaussian random bases have dense pairwise angles.
\end{itemize}
}

%

{\color{black} It is worth mentioning that SR 
	is emerging as a new building block in 
	lattice-reduction-aided MIMO detection.
	Thus, the proposed SR variants may also benefit list sphere decoding \cite{hochwald2003achieving} and Klein's sampling algorithm \cite{Liu2011it}.
}

The rest of this paper is organized as follows. Backgrounds about
lattices and lattice reduction in MIMO are reviewed in Section II.
The SR framework is subsequently introduced in Section III. The low-complexity
version of SR based on hashing is given in Section IV. After that,
Section V presents the simulation results. Conclusions and possible future research are presented in the last section.

Notation: Matrices and column vectors are denoted by uppercase and
lowercase boldface letters. The $i$th column and $\left(j,i\right)$th
entry of $\mathbf{B}$ are respectively denoted as $\mathbf{b}_{i}$
and $b_{i,j}$. $\mathbf{I}_{n}$ and $\mathbf{0}_{n}$ respectively
denote the $n\times n$ identity matrix and $n\times1$ zero vector,
and the operation $(\cdot)^{\top}$denotes matrix transposition. $\left[n\right]$
denotes the set $\left\{ 1,\ldots\thinspace,n\right\} $. For a set
$\Gamma$, $\mathbf{B}_{\Gamma}$ denotes the columns of $\mathbf{B}$
indexed by $\Gamma$. $\mathrm{span}(\mathbf{B}_{\Gamma})$ denotes
the vector space spanned by vectors in $\mathbf{B}_{\Gamma}$. $\pi_{\mathbf{B}_{\Gamma}}(\mathbf{x})$
and $\pi_{\mathbf{B}_{\Gamma}}^{\bot}(\mathbf{x})$ denote the projection
of $\mathbf{x}$ onto $\mathrm{span}(\mathbf{B}_{\Gamma})$ and the
orthogonal complement of $\mathrm{span}(\mathbf{B}_{\Gamma})$, respectively.
$\lfloor x\rceil$ denotes rounding $x$ to the nearest integer, $|x|$
denotes getting the absolute value of $x$, and $\left\Vert \mathbf{x}\right\Vert $
denote the Euclidean norm of vector $\mathbf{x}$. $\mathbb{N}$ and
$\mathbb{Z}$ respectively denotes the set of natural numbers and
integers. The set of $n\times n$ integer matrices with determinants
$\pm1$ is denoted by $\mathrm{GL}_{n}(\mathbb{Z})$.

\section{Preliminaries}

\subsection{Lattices}

An $n$-dimensional lattice $\Lambda$ is a discrete additive subgroup
in the real field $\mathbb{R}^{n}$. Similarly to the fact that any
finite-dimensional vector space has a basis, a lattice has a basis.
To consider a square matrix for simplicity, a lattice generated by
basis $\mathbf{B}=[\mathbf{b}_{1},...,\mathbf{b}_{n}]\in\mathbb{R}^{n\times n}$
is defined as
\[
\Lambda(\mathbf{B})=\left\{ \mathbf{v}\mid\mathbf{v}=\sum_{i\in[n]}c_{i}\mathbf{b}_{i}\thinspace;\thinspace c_{i}\in\mathbb{Z}\right\} .
\]
The dual lattice of $\Lambda$ is defined as $\Lambda^{\dagger}=\left\{ \mathbf{u}\in\mathbb{R}^{n}\mid\langle\mathbf{u},\mathbf{v}\rangle\in\mathbb{Z},\forall\mathbf{v}\in\Lambda\right\} $.
One basis of $\Lambda^{\dagger}$ is given by $\mathbf{B}^{-\top}$.

The Gram-Schmidt (GS) basis of $\mathbf{B}$, referred to as $\mathbf{B}^{*}$,
is found by using $\mathbf{b}_{i}^{*}=\pi_{\left\{ \mathbf{b}_{1}^{*},...,\mathbf{b}_{i-1}^{*}\right\} }(\mathbf{b}_{i})=\mathbf{b}_{i}-\sum_{j=1}^{i-1}\mu_{i,j}\mathbf{b}_{j}^{*}$,
where $\mu_{i,j}=\langle\mathbf{b}_{i},\mathbf{b}_{j}^{*}\rangle/||\mathbf{b}_{j}^{*}||^{2}$.

The $i$th successive minimum of an $n$ dimensional lattice $\Lambda(\mathbf{B})$
is the smallest real positive number $r$ such that $\Lambda$ contains
$i$ linearly independent vectors of length at most $r$: 
\[
\lambda_{i}(\mathbf{B})=\inf\left\{ r\mid\dim(\mathrm{span}((\Lambda\cap\mathcal{B}(\mathbf{0},r)))\geq i\right\} ,
\]
in which $\mathcal{B}(\mathbf{t},r)$ denotes a ball centered at $\mathbf{t}$
with radius $r$. 

The orthogonality defect (OD) can alternatively quantify the goodness
of a basis 
\begin{equation}
\eta(\mathbf{B})=\frac{\prod_{i=1}^{n}||\mathbf{b}_{i}||}{\sqrt{|\mathrm{det}(\mathbf{B}^{T}\mathbf{B})}|}.\label{eq:OD}
\end{equation}
From Hadamard's inequality, we know that $\eta(\mathbf{B})\geq1$.
As the determinant of a given basis is fixed, the parameter is proportional
to the product of the lengths of the basis vectors. A necessary condition
for reaching the smallest orthogonality defect is to have a short
basis length defined as $l(\mathbf{B})=\max_{i}\left\Vert \mathbf{b}_{i}\right\Vert $. 

The $\varepsilon$CVP problem is, given a vector $\mathbf{y}\in\mathbb{R}^{n}$
and a lattice $\Lambda(\mathbf{B})$, find a vector $\mathbf{v}\in\Lambda(\mathbf{B})$
such that: 
\[
\left\Vert \mathbf{y}-\mathbf{v}\right\Vert ^{2}\leq\varepsilon\left\Vert \mathbf{y}-\mathbf{w}\right\Vert ^{2},\thinspace\thinspace\forall\thinspace\mathbf{w}\in\Lambda(\mathbf{B}).
\]
An algorithm that solves an $\varepsilon$CVP problem is referred
to as an $\varepsilon$CVP oracle. We write $\mathbf{v}=\varepsilon\mathrm{CVP}(\mathbf{y},\mathbf{B})$
or $\mathbf{v}=\mathrm{CVP}(\mathbf{y},\mathbf{B})$ if $\varepsilon=1$.

\subsection{Lattice-reduction-aided MIMO detection}

We considered an uplink multiuser large MIMO system, in which $n_{T}$
single-antenna users send data to a base station with $n_{R}$ antennas,
and both $n_{T},n_{R}$ are in the order of tens or hundreds. A received
complex-valued signal vector at the base station is written as:

\begin{equation}
\mathbf{y}_{c}=\mathbf{B}_{c}\mathbf{x}_{c}+\mathbf{w}_{c},\label{EqMimoCModel}
\end{equation}
where $\mathbf{B}_{c}\in\mathbb{C}^{n_{R}\times n_{T}}$ denotes a
channel matrix perfectly known at the base station, $\mathbf{x}_{c}\in\mathbb{C}^{n_{T}}$
refers to a signal vector with entries drawn from a QAM constellation,
and $\mathbf{w}_{c}\in\mathbb{C}^{n_{R}}$ denotes a zero-mean additive
noise vector with entries independently and identically following
the complex normal distribution $\mathcal{CN}\left(0,\sigma_{w}^{2}\right)$. 

To simplify the analysis we will focus on representations in the real
field, so (\ref{EqMimoCModel}) is transformed to an equivalent real
value system with
\begin{equation}
\mathbf{y}=\bar{\mathbf{B}}\mathbf{x}+\mathbf{w},\label{EqMimoModel}
\end{equation}
where 
\begin{equation}
\bar{\mathbf{B}}=\left[\begin{array}{cc}
\Re(\mathbf{B}_{c}) & -\Im(\mathbf{B}_{c})\\
\Im(\mathbf{B}_{c}) & \Re(\mathbf{B}_{c})
\end{array}\right],\label{EqMatrix}
\end{equation}
and $\mathbf{y}=\left[\Re(\mathbf{y})^{\top},\Im(\mathbf{y})^{\top}\right]^{\top}$,
$\mathbf{x}=\left[\Re(\mathbf{x})^{\top},\Im(\mathbf{x})^{\top}\right]^{\top}$,
$\mathbf{w}=\left[\Re(\mathbf{w})^{\top},\Im(\mathbf{w})^{\top}\right]^{\top}$
are all real and imaginary compositions. Here the noise variance of
$\mathbf{w}$ becomes $\sigma^{2}=\sigma_{w}^{2}/2$.

Lattice reduction is essentially multiplying a given basis with a
unimodular matrix $\mathbf{U}\in\mathrm{GL}_{n}(\mathbb{Z})$ to get
a reduced basis $\tilde{\mathbf{B}}\triangleq \bar{\mathbf{B}}\mathbf{U}$.
For a lattice-reduction-aided detector, we first rewrite Eq. (\ref{EqMimoModel})
as: 
\[
\mathbf{y}=\tilde{\mathbf{B}}(\mathbf{U}^{-1}\mathbf{x})+\mathbf{w}.
\]
To make the unimodular transform compact for the QAM constellation,
we need to scale and shift signal vector $\mathbf{x}$ to get $\mathbf{x}\leftarrow(\mathbf{x}+\mathbf{1}_{n\times1})/2$,
so that the constraint on $\mathbf{x}$ become a consecutive integer
set $\Xi^{n}$. Let $\mathbf{y}\leftarrow(\mathbf{\mathbf{y}}+\tilde{\mathbf{B}}\mathbf{U}^{-1}\mathbf{1}_{n\times1})/2$,
then the inferred signal vector is given by:
\begin{equation}
\hat{\mathbf{x}}=2\mathcal{Q}_{\Xi^{n}}(\mathbf{U}\mathcal{Q}_{\mathbb{Z}^{n}}(\mathcal{E}(\mathbf{y},\tilde{\mathbf{B}})))-\mathbf{1}_{n\times1},\label{eq:sic detec}
\end{equation}
where $\mathcal{E}(\mathbf{y},\tilde{\mathbf{B}})$ denotes a low-complexity
detector that could be zero-forcing (ZF) or successive-interference-cancellation
(SIC), and $Q\left(\cdot\right)$ denotes a quantization function
with respect to its subscript. Given certain information about the
signal vector, the detectors can be implemented under an minimum-mean-square-error
(MMSE) principle. The MMSE-based ZF/SIC detectors are similarly given
by extending the size of the system: $\mathbf{y}\leftarrow[\mathbf{y}^{\top},\mathbf{0}_{1\times n}]^{\top}$,
$\bar{\mathbf{B}}\leftarrow[\bar{\mathbf{B}}^{\top},\sigma/\sigma_{s}\mathbf{I}_{n}]^{\top}$,
with $\sigma_{s}^{2}$ referring to the variance of a signal symbol.

\subsection{The objective in lattice reduction}

Hereby we explain the design criteria of lattice reduction used in
MIMO detection. For a set of linearly independent vectors $\bar{\mathbf{B}}=[\bar{\mathbf{b}}_{1},\ldots,\bar{\mathbf{b}}_{n}]$,
we define its fundamental parallelepiped as 
\[
\mathcal{P}(\bar{\mathbf{B}})=\left\{ \sum_{i=1}^{n}c_{i}\bar{\mathbf{b}}_{i}\thinspace|\thinspace-1/2\leq c_{i}\leq1/2\right\} .
\]
Choosing $\mathcal{E}(\mathbf{y},\tilde{\mathbf{B}})$ as the SIC
\cite{Ling2011} detector, then the pairwise error probability $P_{e}$
based on (\ref{eq:sic detec}) becomes 
\begin{eqnarray}
P_{e} & = & 1-\mathrm{Pr}(\mathbf{w}\in\mathcal{P}(\bar{\mathbf{B}}^{*}))\nonumber \\
 & = & 1-\prod_{i=1}^{n}\mathrm{Pr}(|\mathbf{w}^{\top}\bar{\mathbf{b}}_{i}^{*}|<\left\Vert \bar{\mathbf{b}}_{i}^{*}\right\Vert ^{2}/2)\nonumber \\
 & = & 1-\prod_{i=1}^{n}\mathrm{erf}\Big(\frac{\left\Vert \bar{\mathbf{b}}_{i}^{*}\right\Vert }{2\sqrt{2}\sigma}\Big)\nonumber \\
 & \leq & 1-\prod_{i=1}^{n}\mathrm{erf}\Big(\frac{1}{2\sqrt{2}\sigma\left\Vert \mathbf{d}_{i}\right\Vert }\Big)\label{eq:pe last}
\end{eqnarray}
where the last inequality comes from $\left\Vert \bar{\mathbf{b}}_{i}^{*}\right\Vert =\left\Vert \pi_{\bar{\mathbf{b}}_{1},...,\bar{\mathbf{b}}_{i-1}}(\bar{\mathbf{b}}_{i})\right\Vert \geq\left\Vert \pi_{\bar{\mathbf{b}}_{1},...,\bar{\mathbf{b}}_{i-1},\bar{\mathbf{b}}_{i},...,\bar{\mathbf{b}}_{n}}(\bar{\mathbf{b}}_{i})\right\Vert =1/\left\Vert \mathbf{d}_{i}\right\Vert $,
and $\mathbf{d}_{i}$ is the $i$th vector in the dual basis of $\Lambda^{\dagger}$.

From (\ref{eq:pe last}), it becomes clear that the upper bound on
$P_{e}$ is mainly controlled by the lengths of vectors in the dual basis, i.e.,  $\left\Vert \mathbf{d}_{1}\right\Vert ,\ldots,\left\Vert \mathbf{d}_{n}\right\Vert $.
Based on this observation, we can solve/approximate the following problem in the dual lattice to attain better error rate performance
for the above lattice-reduction-aided SIC detector.
\begin{defn}[SBP]
The smallest basis problem (SBP) is, given a lattice $\Lambda$, find the
basis with the smallest orthogonality defect.
\end{defn}

To address SBP, a designed reduction algorithm should make all basis
vectors as short as possible. Moreover, since the basis dimension
is in the order of tens or hundreds in large MIMO, we need a low-complexity
lattice reduction algorithm that reduces the basis with satisfactory
performance.

\section{Sequential reduction framework}

\textcolor{black}{The fundamental principle of sequential reduction is to reduce a basis
	vector by using all other vectors that span a sublattice.}
In the new method, given an input basis $\mathbf{B}$\footnote{\textcolor{black}{Unless otherwise specified, $\mathbf{B}$ is chosen from the dual of a channel matrix.}}, we sequentially solve $\mathbf{s}_{i}=\varepsilon\mathrm{CVP}(\mathbf{b}_{i},\mathbf{B}_{\left[n\right]\backslash i})$
with $\left[n\right]\backslash i=\{1,...,n\}\setminus i$. For each
$\mathbf{s}_{i}$, we test whether the residue distance is shorter:
$||\mathbf{b}_{i}-\mathbf{s}_{i}||^2<\tau||\mathbf{b}_{i}||^2$, where $\tau \in (0,1]$ \footnote{ \textcolor{black}{Choosing $\tau>1$ may make the algorithm diverge.}} is a parameter to control the complexity. If this
holds, we update $\mathbf{b}_{i}$ by $\mathbf{b}_{i}\leftarrow\mathbf{b}_{i}-\mathbf{s}_{i}$.
Here both $\mathbf{s}_{i}=\mathbf{0}$ and the $\mathbf{s}_{i}$ that
makes $\left\Vert \mathbf{b}_{i}-\mathbf{s}_{i}\right\Vert =\left\Vert \mathbf{b}_{i}\right\Vert $
are declared as ineffective attempts. A threshold parameter $m$ is
set to count these useless trials. The algorithm terminates if $m>n$,
which means no more vectors can be further reduced. The general form
of sequential reduction is summarized in Algorithm 1. 

 \begin{algorithm} 
\KwIn{lattice basis $\mathbf{B}=[\mathbf{b}_1,\ldots, \mathbf{b}_n]$, complexity threshold $\tau$;} 
\KwOut{reduced lattice basis $\mathbf{B}$.}  
$i=0$, $m=1$\;  \While{$m \leq n$}
{ 
$i \leftarrow (i~\mathrm{mod}~n)+1$; \Comment{The column index}\;
$ \mathbf{s}_{i} =\varepsilon \mathrm{CVP}(\mathbf{b}_i, \mathbf{B}_{[n]\setminus i})$; \Comment{Exact/approximated CVP solvers}\;
    \If{
$||\mathbf{b}_{i}-\mathbf{s}_{i} ||^2< \tau||\mathbf{b}_{i}||^2$
}
{ $ \mathbf{b}_{i} \leftarrow \mathbf{b}_{i}- \mathbf{s}_{i} $\; 
$m=1$\; 
}    
\Else{$m \leftarrow m +1$\;} 
}
\caption{The general form of an SR algorithm.} 
\end{algorithm}

An SR algorithm maintains a lattice basis due to the following reason.
In round $m$, suppose $\sum_{k\in\left[n\right]\backslash i}c_{k}\mathbf{b}_{k}$
is a valid reduction on $\mathbf{b}_{i}$, then the lattice basis
updating process becomes $\mathbf{B}\leftarrow\mathbf{B}\mathbf{T}^{m}$,
with $T_{k,k}^{m}=1\thinspace\forall\thinspace k\in[n]$, $T_{k,i}^{m}=-c_{k}\thinspace\forall\thinspace k\in\left[n\right]\backslash i$
and all other entries are zeros. Since $\mathbf{T}^{m}$ is an integer
matrix with determinant $1$, $\mathbf{T}^{m}$ is unimodular, and
the composition of the transform matrices from different rounds maintains
a unimodular matrix. 

If an exact CVP oracle is chosen in line 4 of Algorithm 1, we call
the algorithm SR-CVP. By choosing other approximate CVP solvers, we can obtain other variants that have lower complexity. As shown in Fig. \ref{fig1_lambda1-1-1}, SR encompasses SR-CVP, SR-Pair and SR-Hash.
The red box in the figure denotes SR, whose SR-CVP,
SR-Pair and SR-Hash algorithms feature decreasing complexity. Their
analogies in the conventional KZ framework are shown in the black
box.

\begin{figure}[htb]
	\center
	
	\includegraphics[clip,width=0.45\textwidth]{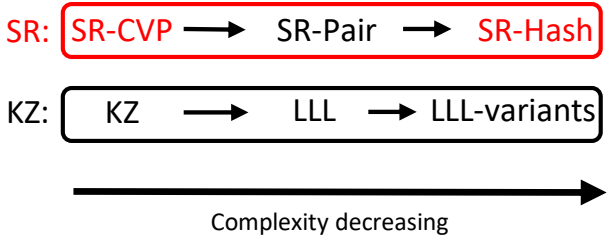}
	
	\caption{SR contains algorithms with different performance-complexity trade-offs.}
	\label{fig1_lambda1-1-1} 
\end{figure}

\subsection{Basis properties of SR-CVP}
\textcolor{black}{We need to understand
	the performance limits of SR by first analyzing SR-CVP.}
Hereby we set $\tau=1$ in the analysis for brevity. When
no more attempts using $\mathbf{s}_{i}=\mathrm{CVP}(\mathbf{b}_{i},\mathbf{B}_{\left[n\right]\backslash i})$
can further reduce the basis in the algorithm, for all $\mathbf{s}_{i}\in\Lambda\left(\mathbf{B}_{\left[n\right]\backslash i}\right)$
we have
\begin{align}
\underset{\mathrm{term}\thinspace1}{\underbrace{\left\Vert \mathbf{b}_{i}\right\Vert ^{2}}} & \leq||\mathbf{b}_{i}-\mathbf{s}_{i}||^{2}\nonumber \\
 & =||\pi_{\mathbf{B}_{\left[n\right]\backslash i}}^{\bot}(\mathbf{b}_{i})+\pi_{\mathbf{B}_{\left[n\right]\backslash i}}(\mathbf{b}_{i})-\mathbf{s}_{i}||^{2}\nonumber \\
 & =\underset{\mathrm{term}\thinspace2}{\underbrace{\left\Vert \pi_{\mathbf{B}_{\left[n\right]\backslash i}}^{\bot}(\mathbf{b}_{i})\right\Vert ^{2}}}+\underset{\mathrm{term}\thinspace3}{\underbrace{\left\Vert \pi_{\mathbf{B}_{\left[n\right]\backslash i}}(\mathbf{b}_{i})-\mathbf{s}_{i}\right\Vert ^{2}}},\label{eq:common equation}
\end{align}
where the second equality is due to Pythagoras' theorem. Note that
SR-CVP provides the tightest constraint on term 3, and approximations
of CVP oracles also distinguish themselves on the same term. Based on (\ref{eq:common equation}), we can prove the following theorem that consists of upper bounds for the basis length and the orthogonality defect.

\begin{thm}
\label{thm:ILD-reduction} For any dimension $n\leq4$, an SR-CVP reduced
basis satisfies:
\begin{eqnarray}
l(\mathbf{B}) & \leq & \sqrt{\frac{4n}{5-n}}\lambda_{n}(\mathbf{B}),\label{eq:ild2}\\
\eta(\mathbf{B}) & \leq & \left(\frac{4}{5-n}\right)^{n/2}\frac{\lambda_{n}^{n}(\mathbf{B})}{\lambda_{1}^{n}(\mathbf{B})}.\label{eq:deltaILD}
\end{eqnarray}
\end{thm}
\begin{proof} We first show upper bounds for terms
2 and 3 in (\ref{eq:common equation}), respectively. By constructing a sublattice
$\Lambda(\mathbf{B}')$ from vectors with lengths $\lambda_{1}(\mathbf{B}),...,\lambda_{n}(\mathbf{B})$
in $\Lambda(\mathbf{B})$, the covering radius satisfies
\begin{align*}
\rho(\mathbf{B}) & =\max_{\mathbf{x}}\mathrm{dist}(\mathbf{x},\Lambda(\mathbf{B}))\\
 & \leq\max_{\mathbf{x}}\mathrm{dist}(\mathbf{x},\Lambda(\mathbf{B}'))\\
 & \leq1/2\sqrt{\sum_{i=1}^{n}\lambda_{i}^{2}(\mathbf{B})},
\end{align*}
where the last inequality is obtained after applying Babai's nearest
plane algorithm \cite{Babai1986}. Since the residue distance of CVP
is upper bounded by the covering radius of the sublattice $\Lambda(\mathbf{B}_{\left[n\right]\backslash i})$,
we have for term 3 that
\begin{align}
\left\Vert \pi_{\mathbf{B}_{\left[n\right]\backslash i}}(\mathbf{b}_{i})-\mathbf{s}_{i}\right\Vert ^{2} & \leq\rho^{2}(\mathbf{B}_{\left[n\right]\backslash i})\nonumber \\
 & \leq\frac{1}{4}\sum_{j\neq i}\lambda_{j}^{2}(\mathbf{B}_{\left[n\right]\backslash i})\nonumber \\
 & \leq\frac{1}{4}\sum_{j\neq i}\left\Vert \mathbf{b}_{j}\right\Vert ^{2}.\label{eq:rho bound}
\end{align}

Regarding term 2, we use QR decomposition to get $\left[\mathbf{B}_{\left[n\right]\backslash i},\mathbf{b}_{i}\right]=\mathbf{QR}$, from which we obtain
 $\left\Vert \pi_{\mathbf{B}_{\left[n\right]\backslash i}}^{\bot}(\mathbf{b}_{i})\right\Vert =|r_{n,n}|$. Now w.l.o.g.
assume that the successive minima $\lambda_{1}(\mathbf{B}),...,\lambda_{n}(\mathbf{B})$
come from vectors $\mathbf{v}_{1}\triangleq\left[\mathbf{B}_{\left[n\right]\backslash i},\mathbf{b}_{i}\right]\mathbf{c}_{1},...,\mathbf{v}_{n}\triangleq\left[\mathbf{B}_{\left[n\right]\backslash i},\mathbf{b}_{i}\right]\mathbf{c}_{n}$.
To produce $n$ linearly independent vectors, there exists at least
one vector denoted as $\mathbf{c}_{k}$ whose $n$th entry $c_{k,n}$
is nonzero. Then we have $\left\Vert \mathbf{R}\mathbf{c}_{k}\right\Vert ^{2}=\lambda_{k}^{2}(\mathbf{B})\leq\lambda_{n}^{2}(\mathbf{B}).$
Together with $\left\Vert \mathbf{R}\mathbf{c}_{k}\right\Vert ^{2}=c_{k,n}^{2}r_{n,n}^{2}+\sum_{j=1}^{n-1}v_{n,j}^{2}\geq\left\Vert \pi_{\mathbf{B}_{\left[n\right]\backslash i}}^{\bot}(\mathbf{b}_{i})\right\Vert ^{2}$,
it arrives at

\begin{align}
 & \left\Vert \pi_{\mathbf{B}_{\left[n\right]\backslash i}}^{\bot}(\mathbf{b}_{i})\right\Vert ^{2}\leq\lambda_{n}^{2}(\mathbf{B}).\label{eq:proofgsbound}
\end{align}
By substituting (\ref{eq:rho bound}) and (\ref{eq:proofgsbound})
 to (\ref{eq:common equation}) for all basis vectors, we have
\[
\begin{cases}
\left\Vert \mathbf{b}_{1}\right\Vert ^{2}\leq\lambda_{n}^{2}(\mathbf{B})+\frac{1}{4}\sum_{j\neq1}\left\Vert \mathbf{b}_{j}\right\Vert ^{2},\\
\vdots\\
\left\Vert \mathbf{b}_{n}\right\Vert ^{2}\leq\lambda_{n}^{2}(\mathbf{B})+\frac{1}{4}\sum_{j\neq n}\left\Vert \mathbf{b}_{j}\right\Vert ^{2}.
\end{cases}
\]
The sum of these $n$ inequalities yields 
\[
\sum_{i=1}^{n}\left\Vert \mathbf{b}_{i}\right\Vert ^{2}\leq n\lambda_{n}^{2}(\mathbf{B})+\frac{n-1}{4}\sum_{i=1}^{n}\left\Vert \mathbf{b}_{i}\right\Vert ^{2}.
\]
If $n\leq4$, we have
\begin{equation}
\sum_{i=1}^{n}\left\Vert \mathbf{b}_{i}\right\Vert ^{2}\leq\frac{4n}{5-n}\lambda_{n}^{2}(\mathbf{B}).\label{eq:fundamentalin}
\end{equation}
Based on Eq. (\ref{eq:fundamentalin}), the longest vector in the
basis can be trivially bounded as
\[
l(\mathbf{B})\leq\sqrt{\frac{4n}{5-n}}\lambda_{n}(\mathbf{B}).
\]
To analyze the orthogonal defect, we apply the arithmetic mean-geometric
mean inequality on (\ref{eq:fundamentalin}) to get 
\begin{equation}
\prod_{i=1}^{n}||\mathbf{b}_{i}||\leq\left(\frac{1}{n}\sum_{i=1}^{n}\left\Vert \mathbf{b}_{i}\right\Vert ^{2}\right)^{n/2}\leq\left(\frac{4}{5-n}\right)^{n/2}\lambda_{n}^{n}(\mathbf{B}).\label{eq:finalod}
\end{equation}
Clearly the volume of the lattice is lower bounded by $\lambda_{1}^{n}\left(\mathbf{B}\right)$
for $n\leq$4, so along with (\ref{eq:finalod}) we obtain (\ref{eq:deltaILD}).

\end{proof} 
\textcolor{black}{
	If we alternatively set $\tau<1$, then upon termination of SR we have  $ ||\mathbf{b}_{i}||^2 \leq 1/\tau ||\mathbf{b}_{i}-\mathbf{s}_{i}||^2 $. Along with the techniques used in Theorem 2, we obtain 
	\begin{equation}
	l(\mathbf{B}) \leq \sqrt{\frac{4n}{4\tau-n+1}}\lambda_{n}(\mathbf{B}). \label{discu1}
	\end{equation}
	Since we have to ensure that the denominator $4\tau-n+1$ is larger than $0$,   we claim that inequality (\ref{discu1}) holds if $n<4\tau + 1$.
	If the CVP oracle is replaced by another suboptimal solver referred to as $\varepsilon$CVP,  then when bounding term 3  we have 
	\[\left\Vert \pi_{\mathbf{B}_{\left[n\right]\backslash i}}(\mathbf{b}_{i})-\mathbf{s}_{i}\right\Vert ^{2}   \leq \varepsilon \rho^{2}(\mathbf{B}_{\left[n\right]\backslash i}).\]
	Similarly to the above, it yields
	\begin{equation}
	l(\mathbf{B}) \leq \sqrt{\frac{4n}{4-\varepsilon n + \varepsilon}}\lambda_{n}(\mathbf{B}), \label{discu2}
	\end{equation}
  in which $n<4/\varepsilon + 1$.
}

Let $\theta_{i}$ be the angle between $\mathbf{b}_{i}$ and the subspace
$\mathrm{span}(\mathbf{B}_{\left[n\right]\backslash i})$,
and define $\theta_{\max}\triangleq\max_{i}\theta_{i}$.
Such a maximum angle between basis vectors and subspaces can also be bounded, as shown in the following theorem.

\begin{thm}
\label{thm:An-SR-CVP-reducedANGLE}An SR-CVP reduced basis satisfies
$\cos^{2}\theta_{\max}\leq\frac{n-1}{4}.$
\end{thm}
\begin{IEEEproof}
Based on (\ref{eq:common equation}) and (\ref{eq:rho bound}) we
have 
\begin{equation}
\left\Vert \mathbf{b}_{i}\right\Vert ^{2}-\left\Vert \pi_{\mathbf{B}_{\left[n\right]\backslash i}}^{\bot}(\mathbf{b}_{i})\right\Vert ^{2}\leq\left\Vert \pi_{\mathbf{B}_{\left[n\right]\backslash i}}(\mathbf{b}_{i})-\mathbf{s}_{i}\right\Vert ^{2}\leq\frac{1}{4}\sum_{j\neq i}\left\Vert \mathbf{b}_{j}\right\Vert ^{2}.\label{eq:common equation 2}
\end{equation}
It then follows from $\left\Vert \mathbf{b}_{i}\right\Vert ^{2}\cos^{2}\theta_{i}=\left\Vert \mathbf{b}_{i}\right\Vert ^{2}-\left\Vert \pi_{\mathbf{B}_{\left[n\right]\backslash i}}^{\bot}(\mathbf{b}_{i})\right\Vert ^{2}$
that 
\[
\left\Vert \mathbf{b}_{i}\right\Vert ^{2}\cos^{2}\theta_{\max}\leq\left\Vert \mathbf{b}_{i}\right\Vert ^{2}\cos^{2}\theta_{i}\leq\frac{1}{4}\sum_{j\neq i}\left\Vert \mathbf{b}_{j}\right\Vert ^{2}.
\]
Similarly to the techniques used in proving Theorem \ref{thm:ILD-reduction},
we sum (\ref{eq:common equation 2}) for $i=1,\ldots,n$ to get 
\[
\cos^{2}\theta_{\max}\leq\frac{n-1}{4}.
\]
\end{IEEEproof}
Clearly the above theorem is non-trivial when $n\leq4$, and this
will come in handy when attacking a counter example in subsection
\ref{subsec:Comparisons-with-weak}.

\subsection{Discussions}
\begin{enumerate}
\item Comparison with $\eta$-Greedy reduction \cite[Fig.5]{Nguyen2012}
(also noted as $\mathrm{ELR}^{+}$-SLV in \cite{Zhou2013a}). Rather
than applying CVP for all vectors, $\eta$-Greedy reduction only performs
CVP for the longest basis vector . According to its definition \cite{Nguyen2012},
it is only a special case of SR-CVP and all SR-CVP reduced basis must
be greedy-reduced. For example, consider the following basis
\[
\left[\begin{array}{ccccc}
2 & 0 & 0 & 0 & 1\\
0 & 2 & 0 & 0 & 1\\
0 & 0 & 2 & 0 & 1\\
0 & 0 & 0 & 2 & 1\\
0 & 0 & 0 & 0 & \varepsilon
\end{array}\right]
\]
with parameter $\varepsilon\in\left(0,1\right)$. {\color{black} The shortest vector $\left[\begin{array}{ccccc}
	0 & 0 & 0 & 0 & \pm 2\varepsilon\end{array}\right]^\top$ cannot be reached by greedy reduction. Specifically, by using $1\times \mathbf{b}_5$ as the query point,  $\eta$-Greedy cannot find
	$2\mathbf{b}_{5}-\mathbf{b}_{1}-\mathbf{b}_{2}-\mathbf{b}_{3}-\mathbf{b}_{4}$ and
	$-2\mathbf{b}_{5}+\mathbf{b}_{1}+\mathbf{b}_{2}+\mathbf{b}_{3}+\mathbf{b}_{4}$. In contrast, SR-CVP additionally considers the cases of using $\mathbf{b}_{1}$, $\mathbf{b}_{2}$, $\mathbf{b}_{3}$, and $\mathbf{b}_{4}$ as query points. A shortest vector   
	$\mathbf{v}=\sum_{i=1}^{n}c_{i}\mathbf{b}_{i}$ with at least one coefficient
	$c_{k}=\pm1$ must be contained in the SR-CVP reduced basis.} 

\item Comparison with KZ and its variants \cite{Lagarias1990,Lyu2017}.
Recall that a basis $\mathbf{B}$ is called KZ reduced if it satisfies
the size reduction conditions, and $\pi_{\mathbf{B}_{\left[i-1\right]}}^{\perp}(\mathbf{b}_{i})$
is the shortest vector of the projected lattice $\pi_{\mathbf{B}_{\left[i-1\right]}}^{\perp}([\mathbf{b}_{i},\ldots\thinspace,\mathbf{b}_{n}])$
for $1\leq i\leq n$ \cite{Lagarias1990}. For a KZ reduced basis,
it satisfies \cite{Lagarias1990} $\left\Vert \mathbf{b}_{i}\right\Vert \leq\frac{\sqrt{i+3}}{2}\lambda_{i}(\mathbf{B}),\thinspace1\leq i\leq n.$
Though boosted KZ \cite{Lyu2017} can solve the length increasing
issue caused by size reduction, tuning $\pi_{\mathbf{B}_{\left[i-1\right]}}^{\perp}(\mathbf{b}_{i})$
to be the shortest vector in the projected lattice can still make
the basis longer. On the contrary, this issue is totally avoided in
SR-CVP.
\item Comparison with Minkowski reduction. Recall that a lattice basis $\mathbf{B}$
is called Minkowski reduced if for any integers $c_{1}$, ..., $c_{n}$
such that $c_{i}$, ..., $c_{n}$ are altogether coprime, it has $\left\Vert \mathbf{b}_{1}c_{1}+\cdots+\mathbf{b}_{n}c_{n}\right\Vert \geq\left\Vert \mathbf{b}_{i}\right\Vert $
for $1\leq i\leq n$ \cite{Zhang2012tsp}. For a Minkowski reduced
basis, it satisfies \cite{Zhang2012tsp} $\left\Vert \mathbf{b}_{i}\right\Vert \leq\max\left\{ 1,(5/4)^{(i-4)/2}\right\} \lambda_{i}(\mathbf{B}),\thinspace1\leq i\leq n.$
Whereas Minkowski reduction is optimal as it reaches all the
successive minima when $n\leq4$, our results in Theorem \ref{thm:ILD-reduction}
only show the SR-CVP reduced basis is not far from the optimal one.
Here we argue that SR-CVP has simpler structure. {\color{black} While
	Minkowski reduction requires solving integer least squares problems with GCD constraints and delicate basis expansion, SR-CVP only involves unconditional CVP solvers and its basis expansion process is trivial. Moreover, the SR-CVP algorithm can be approximately implemented by its many low-complexity siblings in the SR family.}
\end{enumerate}

\subsection{\label{subsec:Complexity-of-SR}Complexity of SR and SR-CVP}

We argue that even when the threshold parameter $\tau=1$, the decrease
from $||\mathbf{b}_{i}||$ to $||\mathbf{b}_{i}-\mathbf{s}_{i}||$
can be finitely counted because a lattice is discrete. Therefore we
define $\epsilon=\sup_{\mathbf{b}_{i},\mathbf{s}_{i}}\frac{||\mathbf{b}_{i}-\mathbf{s}_{i}||}{||\mathbf{b}_{i}||}$
which satisfies $\epsilon<1$. As $\sum_{i=1}^{n}\left\Vert \mathbf{b}_{i}\right\Vert ^{2}$
is no smaller than $\sum_{i=1}^{n}\lambda_{i}^{2}(\mathbf{B})$ while
this metric keeps decreasing for every $n$ iterations, the number
of calls to the CVP oracle is not larger than $n\log(\frac{\left\Vert \mathbf{B}\right\Vert _{\mathrm{F}}^{2}}{\sum_{i=1}^{n}\lambda_{i}^{2}(\mathbf{B})})$
with $\tau\leq1$, where $\left\Vert \mathbf{B}\right\Vert _{\mathrm{F}}^{2}$
denotes the Frobenius norm of the input basis, and the $\log$ function
is over $\min\left\{ 1/\tau,1/\epsilon\right\} $. Therefore we conclude
that the number of iterations in SR is polynomial.

Regarding SR-CVP, since the reduction in each round is quite strong,
we can use the following heuristic implementation to minimize the
number of iterations: \textcolor{black}{first reduce the longest vector (similarly to $\eta$-Greedy), then reduce other basis vectors in descending order with
	$n-1$ rounds of CVP.}  Our simulation results show that this version
of SR-CVP is competitive with Minkowski reduction and boosted KZ reduction.

While we can employ a state-of-the-art implementation for CVP,
its complexity for a random basis is exponential \cite{Jalden2005,AonoN17,DBLP:journals/tit/WangL18}.
In the next section, we will focus on approximate versions of CVP. 

\section{Hash-based Approximation: SR-Hash}

\subsection{The nearest neighbor problem in SR-Pair}

When the $\varepsilon$CVP subroutine is not implemented with an exact
CVP algorithm but rather a pairwise cancellation with the following
form:
\begin{align}
\mathbf{b}_{i} & =\arg\min||\mathbf{b}_{i}^{(j)}||,j=\left\{ 1,\ldots,N\right\} \backslash i,\label{eq:app near}\\
\mathbf{b}_{i}^{(j)} & =\mathbf{b}_{i}-\lfloor\langle\mathbf{b}_{i},\mathbf{b}_{j}\rangle/\langle\mathbf{b}_{j},\mathbf{b}_{j}\rangle\rceil\mathbf{b}_{j},\nonumber 
\end{align}
we refer to the whole algorithm as SR-Pair. This algorithm coincides
with the element-based reduction in \cite{Zhou2013}. Although this
sub-routine only has a complexity in the order of $O(nN)$, reaching
another variant with lower complexity is possible.

Recall the nearest neighbor problem in the field of large dimensional
data processing is: given a list of $n$-dimensional vectors $L=\{\mathbf{v}_{1},\mathbf{v}_{2},\ldots,\mathbf{v}_{N}\}\in\mathbb{R}^{n}$,
preprocess $L$ in such a way that, when later given a target vector
$\mathbf{q}\notin L$, one can efficiently find an element $\mathbf{v}\in L$
which is almost the closest to $\mathbf{q}$. Since Eq. (\ref{eq:app near})
exactly defines a search for the nearest neighbor of $\mathbf{b}_{i}$
among the vectors in $\mathbf{B}$, then it motivates us to reduce
this complexity to $O(n\log N)$ based on locality-sensitive-hashing
(LSH) \cite{Gionis1999,Andoni2006}. 
\begin{rem}
If we choose SIC as the $\varepsilon$CVP subroutine, then along with
LLL preprocessing we have \cite{Babai1986}
\begin{equation}
\left\Vert \mathbf{b}_{i}-\pi_{\mathbf{B}_{\left[n\right]\backslash i}}(\mathbf{b}_{i})-\mathbf{s}_{i}\right\Vert \leq2(2/\sqrt{3})^{n-1}\rho(\mathbf{B})\label{eq:lovasz bound}
\end{equation}
for such an SR-SIC algorithm. However, the computation complexity
of this algorithm is still too high as it requires the pre-processing
by LLL. 
\end{rem}

\subsection{Angular LSH}

LSH roughly works as follows: first all $N$ candidates are dispatched
to different buckets with labels, then when searching the nearest
neighbor of a query point  $\mathbf{q}$, we can alternatively do
this only for $N'$ candidates that have the same label with $\mathbf{q}$,
where $N'\ll N$. There are label functions $f$ which map an $n$-dimensional
vector $\mathbf{v}$ to a low-dimensional sketch of $\mathbf{v}$.
For certain distance function $D$, vectors which are nearby in the
sense of $D$ have a high probability of having the same sketch, while
vectors which are far away have a low probability of having the same
image under $f$. 

To reach this property, we introduce the definition of an LSH family
$\mathcal{F}$. 
\begin{defn}
A family $\mathcal{F}=\{f:\mathbb{R}^{n}\rightarrow\mathbb{N}\}$
of hash functions is said to be $(r_{1},r_{2},p_{1},p_{2})$-sensitive
for a similarity measure $D$ if for any $\mathbf{u},\mathbf{v}\in\mathbb{R}^{n}$,
we have i) If $D(\mathbf{u},\mathbf{v})\leq r_{1}$, then $\mathrm{Pr}_{f\in\mathcal{F}}(f(\mathbf{u})=f(\mathbf{v}))\geq p_{1}$;
ii)If $D(\mathbf{u},\mathbf{v})\geq r_{2}$, then $\mathrm{Pr}_{f\in\mathcal{F}}(f(\mathbf{u})=f(\mathbf{v}))\leq p_{2}$. 

For the sake of constructing a hash family with $p_{1}\approx1$ and
$p_{2}\approx0$, normally one first constructs $p_{1}\approx p_{2}$
and then uses the so called AND- and OR-compositions to turn it into
an $(r_{1},r_{2},p_{1}',p_{2}')$-sensitive hash family $\mathcal{F}'$
with $p_{1}'>p_{1}$ and $p_{2}'<p_{2}$, thereby amplifying the gap
between $p_{1}$ and $p_{2}$. Specifically, by combining $k$ AND-compositions
and $t$ OR-compositions, we can turn an $(r_{1},r_{2},p_{1},p_{2})$-sensitive
hash family $\mathcal{F}$ into an $(r_{1},r_{2},1-\left(1-p_{1}^{k}\right)^{t},1-\left(1-p_{2}^{k}\right)^{t})$-sensitive
hash family $\mathcal{F}'$. As long as $p_{1}>p_{2}$, we can always
find values of $k$ and $t$ such that $1-\left(1-p_{1}^{k}\right)^{t}\rightarrow1$
and $1-\left(1-p_{2}^{k}\right)^{t}\rightarrow0$. 

Note that if given a hash family $\mathcal{H}$ which is $(r_{1},r_{2},p_{1},p_{2})$-sensitive
with $p_{1}\gg p_{2}$, then we can use $\mathcal{F}$ to distinguish
between vectors which are at most $r_{1}$ away from $\mathbf{v}$,
and vectors which are at least $r_{2}$ away from $\mathbf{v}$ with
non-negligible probability, by only looking at their hash values.
Although large values of $k$ and $t$ can amplify the gap between
$p_{1}$ and $p_{2}$, large parameters come at the cost of having
to compute many hashes and having to store many hash tables in memory.
To minimize the overall time complexity, we need the following lemma
that shows how to balance $k$ and $t$. In practice, we can further
tune $k$ and $t$ to have the best performance.
\end{defn}
\begin{lem}[\cite{IndykM98,Laarhoven2015}]
Suppose there exists an $(r_{1},r_{2},p_{1},p_{2})$-sensitive family
$\mathcal{F}$. For a list $L$ of size $N$, let 
\[
\rho=\frac{\log p_{1}^{-1}}{\log p_{2}^{-1}},\thinspace k=\frac{\log N}{\log p_{2}^{-1}},\thinspace t=O(N^{\rho}).
\]
 Then given a query point $\mathbf{q}$, with high probability we
can either find an element $\mathbf{v}\in L$ such that $D(\mathbf{q},\mathbf{v})\leq r_{2}$,
or conclude that with high probability, no element $\mathbf{v}\in L$
with $D(\mathbf{q},\mathbf{v})>r_{1}$ exist, with the following costs:
i) Time for preprocessing the list: $O(kN^{1+\rho})$; ii) Space complexity
of the preprocessed data: $O(N^{1+\rho})$; iii) Time for answering
a query: $O(N^{\rho})$.
\end{lem}
In the sequel, we examine the implementation of LSH based on angular
hashing. Angular hashing means generating random hyperplanes $\mathbf{h}_{1},\ldots,\mathbf{h}_{k}$,
such that the whole space is sliced into $2^{k}$ regions. After that,
to find the nearest neighbor of $\mathbf{q}$, one only compares $\mathbf{q}$
to points in the same region $\mathcal{R}$. Here we introduce the
angular distance similarity function

\[
D(\mathbf{u},\mathbf{v})=\arccos\left(\frac{\mathbf{u}^{\top}\mathbf{v}}{\left\Vert \mathbf{u}\right\Vert \left\Vert \mathbf{v}\right\Vert }\right).
\]
 With this measure two vectors are nearby if their common angle is
small. Its corresponding hash family is defined by
\[
\mathcal{F}=\left\{ f_{\mathbf{a}}:\thinspace\mathbf{a}\in\mathbb{R}^{n},\left\Vert \mathbf{a}\right\Vert =1\right\} ,f_{\mathbf{a}}\left(\mathbf{v}\right)=\begin{cases}
1 & \mathrm{if}\thinspace\mathbf{a}^{\top}\mathbf{v}\geq0;\\
0 & \mathrm{if}\thinspace\mathbf{a}^{\top}\mathbf{v}<0.
\end{cases}
\]
 Intuitively, the space that is orthogonal to $\mathbf{a}$ defines
a hyperplane, and $f_{\mathbf{a}}$ maps the two regions separated
by this hyperplane to different bits. In particular, for any two angles
$\theta_{1}<\theta_{2}$, the family $\mathcal{F}$ is $(\theta_{1},\theta_{2},1-\frac{\theta_{1}}{\pi},1-\frac{\theta_{2}}{\pi})$-sensitive.
Further with $k$ AND- and $t$ OR- compositions, we have $(\theta_{1},\theta_{2},1-\left(1-\left(1-\frac{\theta_{1}}{\pi}\right)^{k}\right)^{t},1-\left(1-\left(1-\frac{\theta_{2}}{\pi}\right)^{k}\right)^{t})$-sensitive
hash family.

To illustrate LSH and in particular the angular LSH method described
above, Fig. \ref{fig: lshdemo} shows how hyperplane hashing might
work in a 2-D setting. In the figure, we have a list of $8$ candidates:
$L=\left\{ \mathbf{b}_{1},\ldots,\mathbf{b}_{8}\right\} $, and we
use $k=2$ hyperplanes for $t=2$ hash tables. Each table stores the
hash keys (labels) along with elements being placed in buckets, where
elements having the same keys will be placed in the same buckets.
In the two tables, the AND-compostions of $11$ respectively correspond
to $\mathbf{b}_{1},\mathbf{b}_{2},\mathbf{b}_{3}$ and $\mathbf{b}_{1}$.
Based on OR-composition, the nearest neighbor of $\mathbf{b}_{1}$
is found inside $\left\{ \mathbf{b}_{2},\mathbf{b}_{3}\right\} $.

\begin{figure}[tbh]
\center

\includegraphics[width=3in]{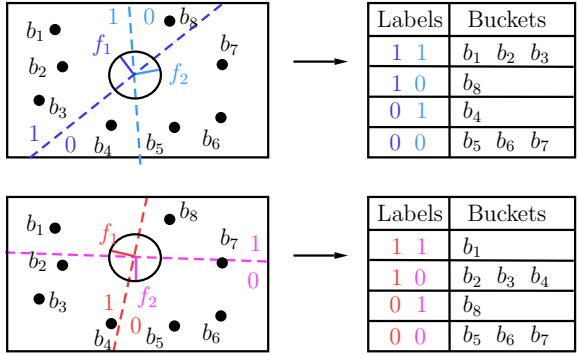}

\caption{Demonstration of LSH.}
\label{fig: lshdemo}
\end{figure}

\subsection{LSH-Based Reduction}

Now we show how to incorporate LSH into the sequential reduction algorithm.
The pseudo-codes of SR-Hash are presented in Algorithm \ref{algoLSHELR}.
It first hashes all vectors in lines 2-3. Then inside the loop of
the element-based reduction, the search for finding the nearest neighbor
of $\mathbf{b}_{i}$ is within $\mathbf{C}=\cup_{l=1}^{t}T_{l}\left(f_{l}\left(\pm\mathbf{b}_{i}\right)\right)$.
Every time when a shorter $\mathbf{b}_{i}$ is found, its hash labels
and positions in buckets are updated in lines 11 and 13.

In summary, SR-Hash can be seen as a generalization of the naive brute-force
search inside SR-Pair for finding nearest neighbors, as $k=0$, $t=1$
corresponds to checking all other basis vectors for nearby vectors,
while increasing both $k$ and $t$ leads to fewer comparisons but
a higher cost of computing hash keys and checking buckets.

   \begin{algorithm}  	\KwIn{original lattice basis $\mathbf{B}=[\mathbf{b}_1,\ldots, \mathbf{b}_n]$, complexity threshold   $\tau$, LSH parameters $t$, $k$.}  	 \KwOut{Reduced basis $[\mathbf{b}_{1},\ldots,\mathbf{b}_{n}]$ of   $\Lambda$.}   	 Initialize $t$ empty hash tables $\left(T_l\right)_{l=1}^{t}$, each has $k$ random hash functions  $f_{l,1},\ldots f_{l,k} \in \mathcal{F}$\; 	 \For{$i=1\cdots n$} 	 {
Add $\mathbf{b}_i$ to all hash tables $\left(T_l\right)_{l=1}^{t}$, with hash values $\left(f_l\left(\mathbf{b}_{i}\right)\right)_{l=1}^{t}$ and vectors in the same bucket noted as $T_l\left(f_l\left(\mathbf{b}_{i}\right)\right)$;} 	
$i=0$, $m = 1$\;  
\While{$m \leq n$}{ 
$i \leftarrow (i~\mathrm{mod}~n)+1$; \Comment{The column index}\; 
	Obtain the set of candidates $\mathbf{C}=\cup_{l=1}^{t} T_l\left(f_l\left(\pm \mathbf{b}_{i}\right)\right)$\; 	 	$\mathbf{c}_{l}=\arg\min_{\mathbf{c}_{l}\in \mathbf{C}}\left\Vert {\mathbf{b}_{i}}-\lfloor\langle \mathbf{b}_{i},\mathbf{c}_{l}\rangle/\langle\mathbf{c}_{l},\mathbf{c}_{l}\rangle\rceil \mathbf{c}_{l}\right\Vert ^{2}$\;  	
$\mathbf{s}_i=\lfloor\langle \mathbf{b}_{i},\mathbf{c}_{l}\rangle/\langle\mathbf{c}_{l},\mathbf{c}_{l}\rangle\rceil \mathbf{c}_{l}$\;
	    \If{$ ||\mathbf{b}_{i}-\mathbf{s}_i||^2< \tau||\mathbf{b}_{i}||^2$} {  	 	    	Remove $\mathbf{b}_i$ from all hash tables\; 	 	    	$\mathbf{b}_{i} \leftarrow \mathbf{b}_{i}-\mathbf{s}_i$\; 	 	    	Add $\mathbf{b}_i$ to all hash tables\; 	 	    		  $m=1$\;  }       	    		  \Else{$m \leftarrow m +1$\;} }           \caption{The SR-Hash   algorithm.}        \label{algoLSHELR}     \end{algorithm}

\subsection{\label{subsec:Comparisons-with-weak}Discussions}
\begin{enumerate}
\item Comparison with SR-CVP. Here we emphasize that SR-Pair/SR-Hash is
only a weak approximation for SR-CVP, and these low complexity algorithms
may have quite inferior performance. Consider the counter example
given for ELR \cite{Zhou2013} (the same as SR-Pair). Clearly SR-Pair/SR-Hash
is unable to reduce a basis whose Gram matrix is
\begin{eqnarray*}
\mathbf{G} & = & \left[\begin{array}{ccc}
1 & 0.5-\nu & 0.5-\nu\\
0.5-\nu & 1 & -0.5+\nu\\
0.5-\nu & -0.5+\nu & 1
\end{array}\right]
\end{eqnarray*}
with $\nu\rightarrow0$. Under spherical coordinate system of $(r,\varrho,\varphi)$
with $r=1$, $\varrho=\pi/3$, and $\varphi=\pi/2-\nu$, the lattice
basis $\mathbf{A}$ corresponded to $\mathbf{G}$ (up to a unitary
transform) is given by
\begin{equation}
\mathbf{A}=\left[\begin{array}{ccc}
\sin\varphi\cos(\pi/3) & -\sin\varphi\cos(\pi/3) & 1\\
\sin\varphi\sin(\pi/3) & \sin\varphi\sin(\pi/3) & 0\\
\cos\varphi & -\cos\varphi & 0
\end{array}\right].\label{eq:badmatrix}
\end{equation}
This basis has an angle $\theta_{i}<\nu\rightarrow0$ between any
$\mathbf{a}_{i}$ and $\mathrm{span}(\mathbf{A}_{[3]\backslash i})$,
and $\eta(\mathbf{A})= \infty $   if $\nu\rightarrow0$.
If $\bm{A}$ is reduced by using SR-CVP, we have $\theta_{\max}\geq\pi/4$
according to Theorem \ref{thm:An-SR-CVP-reducedANGLE}, so $\mathbf{A}$
is not a stable basis for SR-CVP. Moreover, the actual reduced basis
has the following form:
\[
\tilde{\mathbf{A}}=\left[\begin{array}{ccc}
2\sin\varphi\cos(\pi/3)-1 & -\sin\varphi\cos(\pi/3) & 1\\
0 & \sin\varphi\sin(\pi/3) & 0\\
2\cos\varphi & -\cos\varphi & 0
\end{array}\right].
\]
Its OD is 
\[
\eta(\tilde{\mathbf{A}})=\frac{\sqrt{4\cos^{2}\varphi+\sin^{2}\varphi-2\sin\varphi+1}}{\sqrt{3}\sin\varphi\cos\varphi};
\]
when given $\varphi=\pi/2-10^{-4}$, we have $\eta(\tilde{\mathbf{A}})|_{\varphi=\pi/2-10^{-4}}=1.1547$.
\textit{Therefore, the proposed low-complexity SR algorithms are only
feasible for bases whose input vectors are dense in some directions}.
Our simulation results and Appendix \ref{sec:Properties-of-element-based}
will show that the dual lattice basis in MIMO detection is one example
of this.
\item Comparison with LLL and its variants \cite{Ling2010,Wen2014b,DBLP:journals/icl/WenC17}.
Note that the worst case complexity of LLL for bases in the real field
is unbounded \cite{Jalden2008}, and the variants that control the
order of swaps or a selective implementation of size reduction cannot
remove this curse. On the contrary, SR variants with a polynomial
time $\varepsilon\mathrm{CVP}$ routine can enjoy the overall polynomial-time
complexity. Regarding performance bounds, LLL and its variants (the
maintains the Siegel condition and size reduction condition) often
have bounds of the form $l\left(\mathbf{B}\right)\leq2^{n-1}\lambda_{n}\left(\mathbf{B}\right)$,
while SR-Pair and SR-Hash are heuristic. 
\item Comparison with Seysen reduction \cite{Seysen1993}. Rather than minimizing
the orthogonality defect of a basis, a metric called Seysen\textquoteright s
measure can reflect whether both the primal and dual bases are short:
$\sum_{i=1}^{n}\left\Vert \mathbf{b}_{i}\right\Vert ^{2}\left\Vert \mathbf{d}_{i}\right\Vert ^{2}$.
Seeking for the global minimum of this metric is extremely hard; when
referring to Seysen\textquoteright s algorithm \cite{Seethaler2007},
it is the one that finds a local minimum of $\sum_{i=1}^{n}\left\Vert \mathbf{b}_{i}\right\Vert ^{2}\left\Vert \mathbf{d}_{i}\right\Vert ^{2}$
without any theoretical performance guarantee. Similarly to SR-Pair,
Seysen\textquoteright s algorithm performs basis updates in a pair-wise
manner:
\begin{align*}
\mathbf{b}_{j} & =\mathbf{b}_{j}+c_{i,j}\mathbf{b}_{i},i\neq j,
\end{align*}
 with $c_{i,j}=\lfloor\frac{1}{2}\left(\frac{\langle\mathbf{d}_{i},\mathbf{d}_{j}\rangle}{\left\Vert \mathbf{d}_{i}\right\Vert ^{2}}-\frac{\langle\mathbf{b}_{i},\mathbf{b}_{j}\rangle}{\left\Vert \mathbf{b}_{i}\right\Vert ^{2}}\right)\rceil$.
Due to the additional inner product calculation in the dual basis,
Seysen's algorithm is more complicated than SR-Pair, and it does not
support the hash-based implementation. Moreover, in large ($\geq35$)
dimensions Seysen's algorithm often halts at a local minimum \cite[P.375]{Seysen1993}.
Since the error rate performance is only controlled by the length
of the dual basis, our empirical results also show that Seysen's algorithm
is not competitive for large dimensions.
\end{enumerate}

\section{\textcolor{black}{Simulation results}}

\subsection{Performance of SR-CVP}
\begin{figure}[tbh]
\center
\subfloat[Primal bases.]{\includegraphics[width=0.25\textwidth,height=0.4\textwidth]{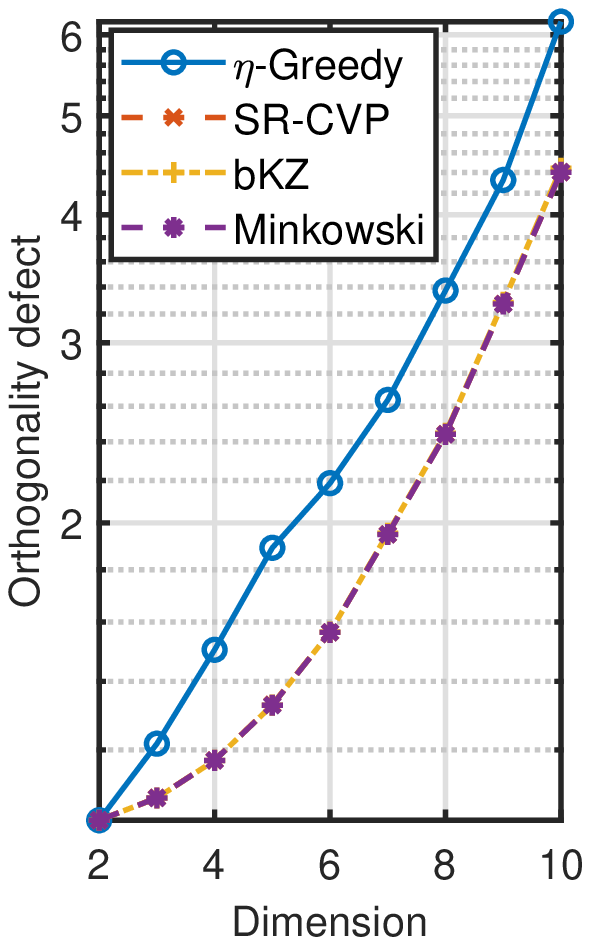}

}\subfloat[Dual bases.]{\includegraphics[width=0.25\textwidth,height=0.4\textwidth]{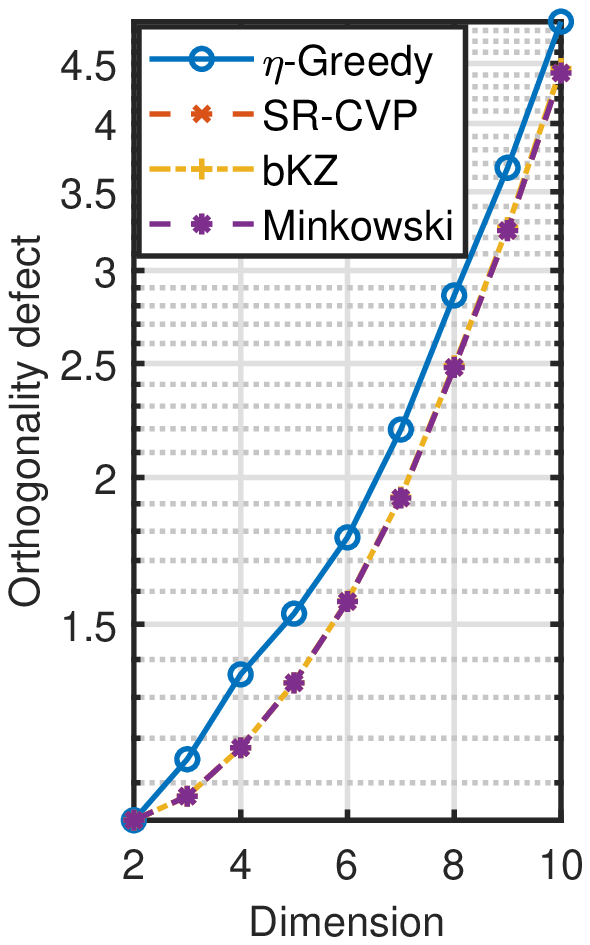}

}\caption{The orthogonal defects of different types of strong reduction.}
\label{fig:odstrong}
\end{figure}

Hereby we employ the OD's to compare SR-CVP with other strong lattice
reduction algorithms, including the boosted Korkin-Zolotarev reduction
noted as ``bKZ'', the Minkowski reduction noted as ``Minkowski'',
and the $\eta$-Greedy reduction \cite[Fig.5]{Nguyen2012} noted as
``$\eta$-Greedy''. Results are averaged over $1\times10^{4}$ Monte-Carlo
runs, and SR-CVP is implemented by the heuristic version in subsection
\ref{subsec:Complexity-of-SR}.

Fig. \ref{fig:odstrong} plots dimension versus OD for distinct algorithms
for the primal and dual of a Gaussian random matrix with entries from
$\mathcal{N}\left(0,1\right)$, respectively. The figure shows that
 ODs of SR-CVP, bKZ and Minkowski reduced bases are almost indistinguishable.
$\eta$-Greedy has the worst performance as expected, because it is
not designed to minimized all basis vectors. Since Minkowski reduction
is the state-of-the-art algorithm for generating the shortest basis
in practice, our results show that SR-CVP  practically reaches
 optimality as well. 

Fig. \ref{fig:odstrongFLOP} plots the averaged number of CVP runs
in Fig. \ref{fig:odstrong} when using $\eta$-Greedy and SR-CVP.
It is known that both Minkowski and bKZ cost around $n$ oracles for
the shortest vector problem (SVP) or CVP. Fig. \ref{fig:odstrongFLOP}
reflects that SR-CVP actually needs fewer than $n$ rounds of CVP,
and is only slightly more complicated than $\eta$-Greedy.

\begin{figure}[tbh]
\center

\subfloat[Primal bases.]{\includegraphics[width=0.25\textwidth,height=0.4\textwidth]{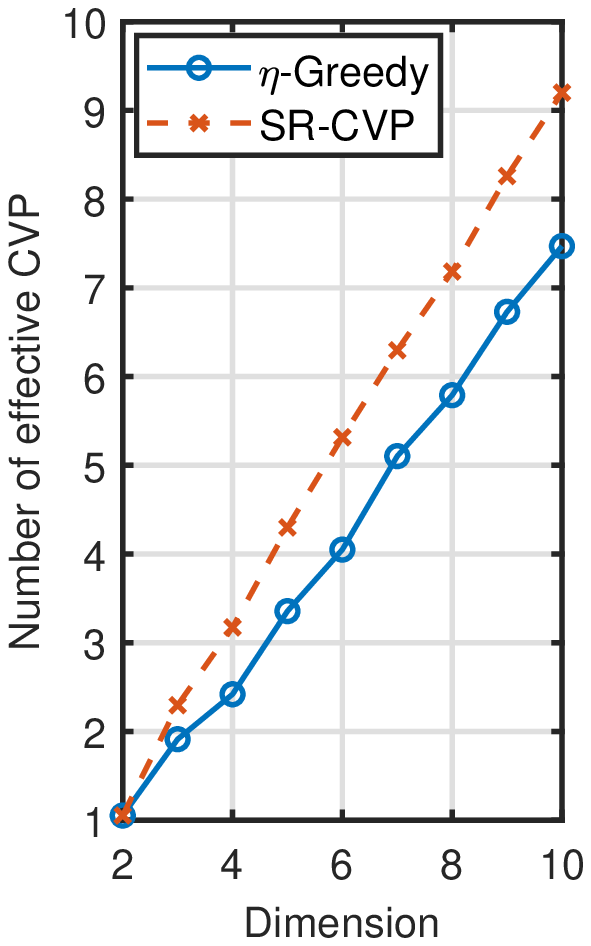}

}\subfloat[Dual bases.]{\includegraphics[width=0.25\textwidth,height=0.4\textwidth]{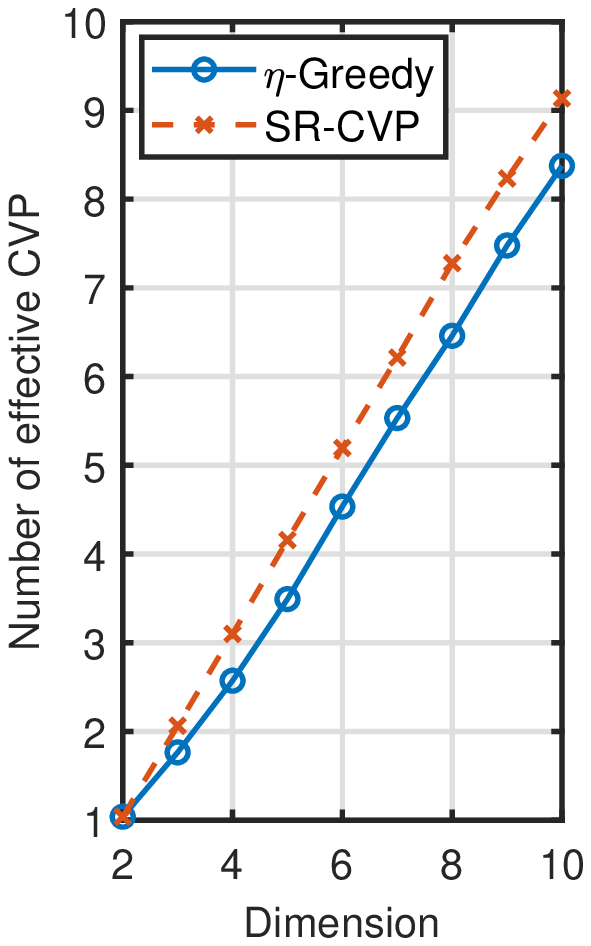}

}\caption{The number of effective CVP runs in $\eta$-Greedy and SR-CVP.}
\label{fig:odstrongFLOP}
\end{figure}

\subsection{SR-Hash vs. SR-Pair and LLL variants}

In this subsection, we study the complexity/performance tradeoffs
of different types of weak lattice reduction. The modulation is set
as $16$ QAM, and the results are obtained from $1\times10^{4}$ Monte
Carlo runs. We denote the zero-forcing detector by ``ZF'',
the successive interference cancellation detector by ``SIC'',  and lattice-reduction-aided detectors with prefixes: ``LLL-SIC/ZF'' \cite{Ling2011}, ``bLLL-SIC/ZF''
\cite{Lyu2017}, ``SR-Pair-SIC/ZF'' (this paper), ``SR-Hash-SIC/ZF'' (this paper), \textcolor{black}{and ``Seysen-SIC/ZF'' \cite{ma2010seysen}. Here comparisons are made {for major lattice-reduction-aided methods} in large-scale MIMO systems, because they represent pre-processing based methods that may attain the diversity order of ML detection \cite{Wuebben2011,Taherzadeh2007b}.}  

\subsubsection{i.i.d. channels}

Assume that each entry of the channel matrix is chosen from a standard
normal distribution $\mathcal{CN}\left(0,1\right)$. Fig. \ref{FIG:BER1}
plots the bit error rate (BER) performance of different uncoded MIMO
detectors in a real domain $2n_{T}\times2n_{R}=60\times60$ MIMO system.
Here the linear detectors are implemented with the MMSE criterion.
Parameters in LSH are chosen as $t=\lfloor n^{0.585}\rceil=11$, $k=\lfloor\log n\rceil=6$. 

In the high SNR region of Fig. \ref{FIG:BER1}-(a), we observe that,
in addition to the well-known fact that ZF, SIC and Seysen-SIC fail to achieve the
full diversity order, b-LLL-SIC, SR-Pair-SIC and SR-Hash-SIC all attain
approximately 1dB gain over LLL-SIC. As for Fig. \ref{FIG:BER1}-(b),
the variants of SR both outperform conventional and boosted LLL algorithms.
Both sub-figures indicate that SR-Hash gets very close to SR-Pair.

\begin{figure}[htb]
\center

\subfloat[Lattice-reduction-aided SIC detectors.]{\includegraphics[width=0.4\textwidth]{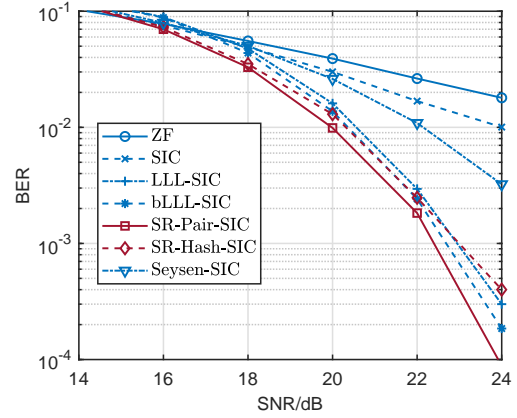}

}

\subfloat[Lattice-reduction-aided ZF detectors.]{\includegraphics[width=0.4\textwidth]{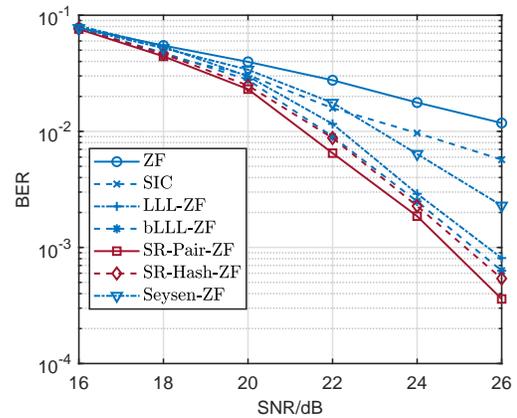}

}

\caption{\textcolor{black}{The BER performance of different detectors in large MIMO.}}
\label{FIG:BER1}
\end{figure}
The complexity of implementing the lattice reduction algorithms is
plotted in Fig. \ref{fig:compexityFig4}, where sub-figure (a) is
for the effective channel matrix under the MMSE criterion, and sub-figure
(b) under the ZF criterion. Considering the difficulty in analyzing
the number of floating-point operations for hash operations, here
we measure the complexity by the number of vector comparisons. This
equals to the number of iterations times: the size of the basis for
SR-Pair, the number of vectors in the same buckets for SR-Hash, and
to the size of vectors for doing size-reductions for both LLL and
bLLL. From Fig. \ref{fig:compexityFig4}-(a), we observe that the
LLL variants are not affected by SNR in the MMSE matrix, and Seysen, SR-Pair and
SR-Hash gradually increase with the rise of SNR. This shows the complexity
of Seysen, SR-Pair and SR-Hash are dependent on the quality of the input bases.
Regarding the stationary lines in Fig. \ref{fig:compexityFig4}-(b),
the numbers of comparisons of Seysen, SR-Pair and SR-Hash reflect the asymptotic
values of their counter-parts in Fig. \ref{fig:compexityFig4}-(a).
Both subfigures reveal that the hash method
helps to reduce the complexity of SR-Pair significantly.
A natural question that arises here is whether the complexity dependency of SR-Pair\&SR-Hash on input bases may lead to inferior performance at low SNR. To address this question, we plot the SNR versus OD relations of different reduction algorithms in Fig. \ref{fig:odmmse}. We observe from the figure that even at low SNR, SR-Pair\&SR-Hash featuring low complexity still outperform Seysen and LLL in terms of OD.

\begin{figure}[htb]
\center

\subfloat[MMSE criterion.]{\includegraphics[width=0.25\textwidth,height=0.4\textwidth]{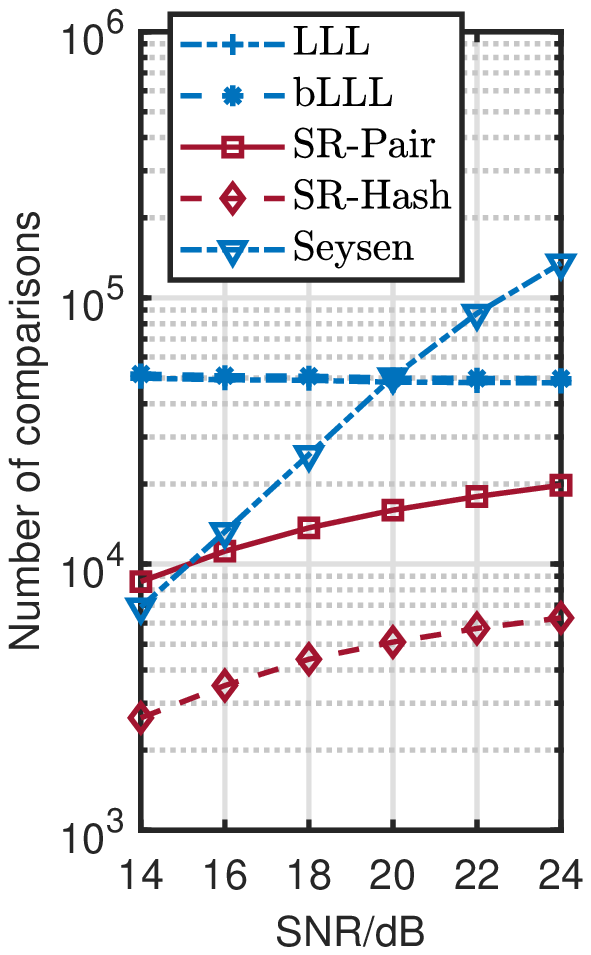}

}\subfloat[ZF criterion.]{\includegraphics[width=0.25\textwidth,height=0.4\textwidth]{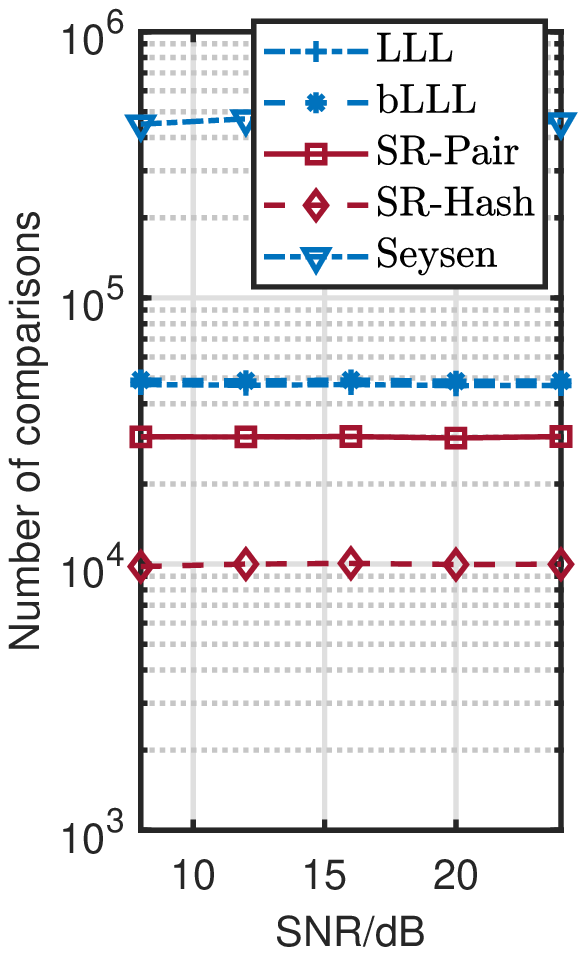}

}\caption{\textcolor{black}{The complexity of different lattice reduction algorithms in i.i.d. channels.}}
\label{fig:compexityFig4}
\end{figure}

 \begin{figure}[htb]
	\center
	
	\includegraphics[width=0.4\textwidth]{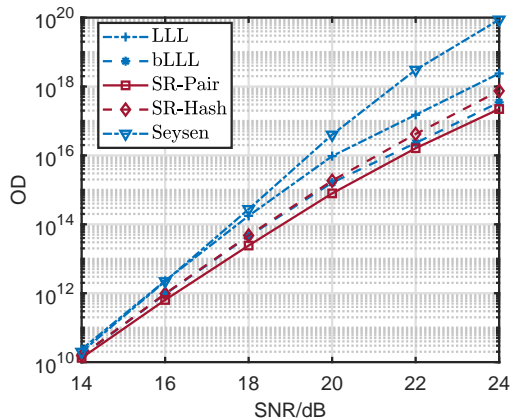}
	
	\caption{\textcolor{black}{The ODs of MMSE matrices reduced by different algorithms.}}
	\label{fig:odmmse}
\end{figure}

\subsubsection{Correlated Channels}

Results in the last example were obtained for i.i.d. frequency-flat
Rayleigh fading channels. The performance of MIMO systems in realistic
radio environments however sometimes depends on spatial correlation.
Therefore, we investigated the effect of channel correlation on the
performance of the new reduction algorithms. Based on \cite{DBLP:journals/tvt/0001KC07},
the spatially correlated channel is modeled as 
\[
\widetilde{\mathbf{B}_{c}}=\Psi\mathbf{B}_{c},
\]
 where $\Psi\in\mathbb{R}^{n_{R}\times n_{R}}$ is the correlation
matrix defined by 
\[
\Psi=\left[\begin{array}{cccc}
1 & \rho & \ldots & \rho^{n_{R}-1}\\
\rho & 1 & \ldots & \rho^{n_{R}-2}\\
\vdots & \vdots & \ddots & \vdots\\
\rho^{n_{R}-1} & \rho^{n_{R}-2} & \ldots & 1
\end{array}\right],
\]
and $\rho$ refers to the spatial correlation coefficient. 

With the same chosen parameters in the algorithm as those for i.i.d.
channels, Fig. \ref{FIG:BER1-1} demonstrates the BER performances
against SNR in correlated channels respectively with $\rho=0.1$ and
$\rho=0.3$. It reveals that, as $\rho$ increases, the SR aided detectors
suffer from more severe performance degradation than the LLL aided
methods, although the BER gaps between SR variants and LLL variants are very small. This is not unexpected because we do have examples showing
SR-Pair/SR-Hash cannot reduce certain matrices (e.g., the matrix in (\ref{eq:badmatrix})).
Lastly, as plotted in Fig. \ref{fig:compexityFig4-1}, the complexity
of SR-Pair/SR-Hash is still much lower than those of LLL variants and Seysen,
and the proposed SR-Hash has much lower complexity than SR-Pair.

\begin{figure}[tbh]
\center

\subfloat[ $\rho=0.1$.]{\includegraphics[width=0.4\textwidth]{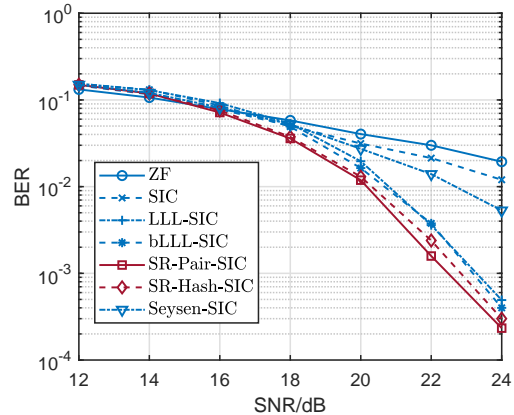}

}

\subfloat[ $\rho=0.3$.]{\includegraphics[width=0.4\textwidth]{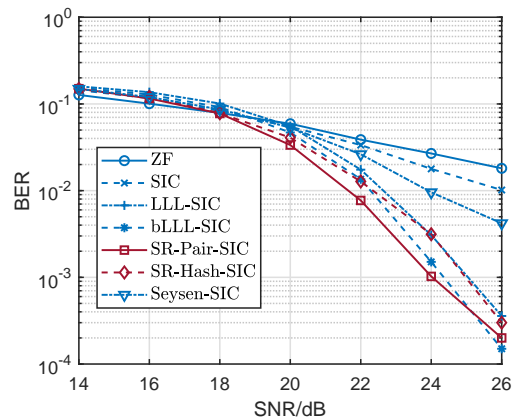}

}

\caption{\textcolor{black}{The BER performance of lattice-reduction-aided SIC detectors in correlated
		channels with $\rho=0.1$ and $\rho=0.3$.} }
\label{FIG:BER1-1}
\end{figure}
\begin{figure}[tbh]
\center

\subfloat[ $\rho=0.1$.]{\includegraphics[width=0.25\textwidth,height=0.4\textwidth]{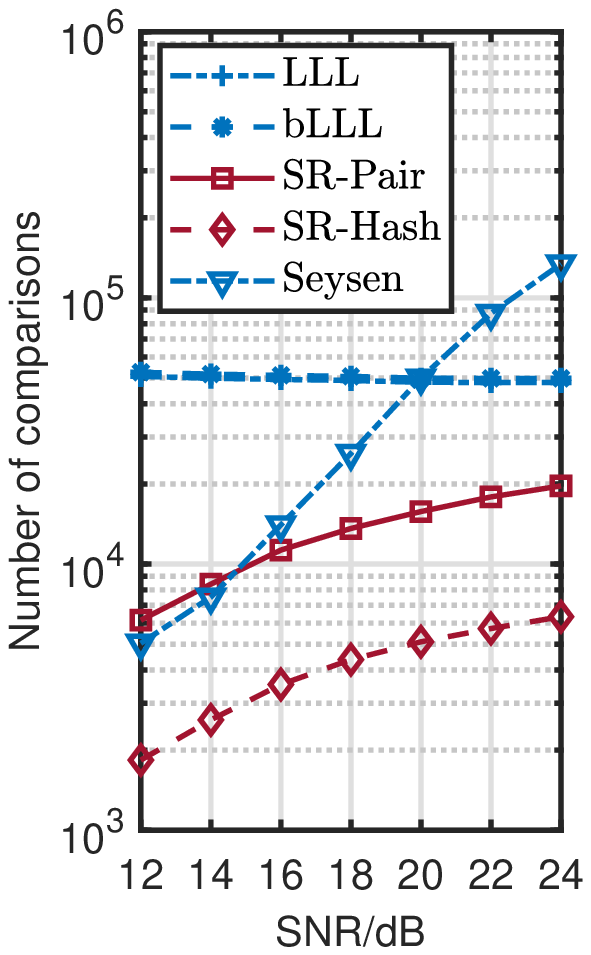}

}\subfloat[ $\rho=0.3$.]{\includegraphics[width=0.25\textwidth,height=0.4\textwidth]{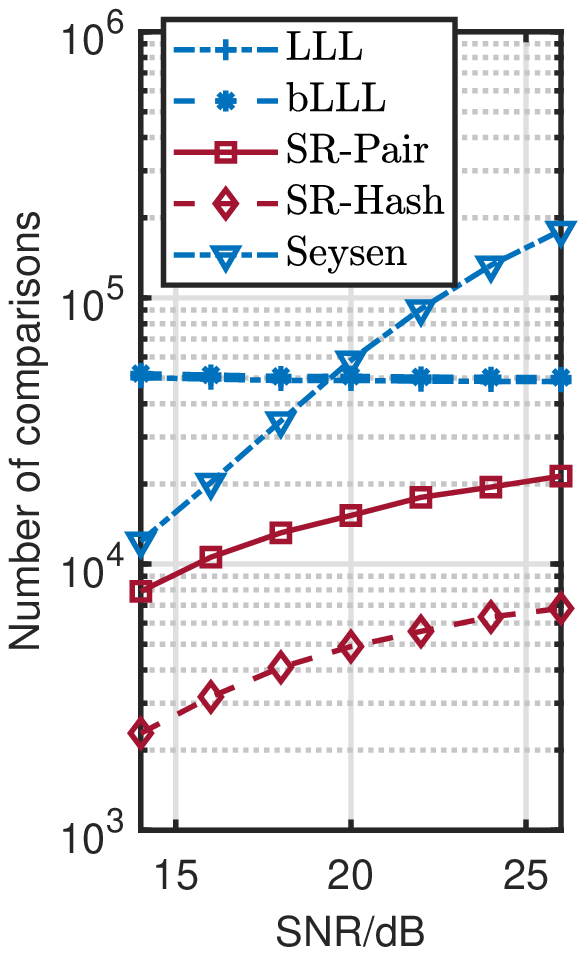}

}\caption{\textcolor{black}{The complexity of lattice reduction algorithms in correlated channels
	with $\rho=0.1$ and $\rho=0.3$.}}
\label{fig:compexityFig4-1}
\end{figure}

\section{\color{black}Conclusions and Future Work}

To summarize, we have unveiled a new lattice reduction family  called sequential reduction, which enjoys a polynomial number of iterations.
Theoretical bounds on basis lengths and orthogonality defects are derived
under the premise that an exact CVP subroutine has been invoked. Though
we only manage to prove these results for small dimensions, they still
provide insights on understanding the performance of such a class
of algorithms. Within the SR framework, the SR-Hash method can serve
as an effective subprogram, and simulation results show that the complexity-performance
trade-off outperforms those of SR-Pair and LLL variants in large MIMO
detection.  

{\color{black} We believe that the studies we initiated here, only scratch
	the tip of the iceberg about the new lattice reduction family. Many important questions remain to be answered. Research on the interactions and combinations of SR with other techniques such as floating-point arithmetic \cite{DBLP:conf/eurocrypt/NguenS05,ChangSV12}, randomized detection algorithms \cite{XieXW19,Liu2011it},  success probability analysis \cite{WenC17,ChangWX13}, and numerous other topics is now being pursued.
}


\appendices{}

\section{\label{sec:Properties-of-element-based}On the type of bases feasible
for SR-Pair\&SR-Hash}

We first argue that a large dimensional Gaussian random basis is always
SR-CVP reduced, and thus being SR-Pair and SR-Hash reduced. The inability
to change such bases is however not a problem because these bases
are close to being orthogonal.

\begin{prop}
\label{prop:nonReduceSLR}For a Gaussian random basis whose entries
follow the distribution $\mathcal{N}\left(0,1\right)$, the probability
that it is not SR-CVP reduced goes to zero as $n\rightarrow\infty$.
\end{prop}

\begin{IEEEproof}
We need to show that for all choices of coefficients $a_{i}'s$ in
$\mathbb{Z}$ with at least one nonzero $a_{i}$, the probability  
\[
\mathrm{Pr}\left(\left\Vert \mathbf{b}_{1}+\sum_{i=2}^{n}\mathbf{b}_{i}a_{i}\right\Vert ^{2}\leq\left\Vert \mathbf{b}_{1}\right\Vert ^{2}\right)
\]
 vanishes as the problem size $n$ increases. Since $\sum_{i=2}^{n}\mathbf{b}_{i}a_{i}$
is an isotropic Gaussian random vector with covariance $\mathbb{E}\left(\left(\sum_{i=2}^{n}\mathbf{b}_{i}a_{i}\right)\left(\sum_{i=2}^{n}\mathbf{b}_{i}a_{i}\right)^{\top}\right)=\left(\sum_{i=2}^{n}a_{i}^{2}\right)\mathbf{I}_{n}$,
then for any $\beta>0$,
\begin{align}
 & \mathrm{Pr}\left(\left\Vert \mathbf{b}_{1}+\sum_{i=2}^{n}\mathbf{b}_{i}a_{i}\right\Vert ^{2}\leq\left\Vert \mathbf{b}_{1}\right\Vert ^{2}\right) \nonumber \\
 & \leq\mathbb{E}\left(e^{-\beta\left(\left\Vert \mathbf{b}_{1}+\sum_{i=2}^{n}\mathbf{b}_{i}a_{i}\right\Vert ^{2}-\left\Vert \mathbf{b}_{1}\right\Vert ^{2}\right)}\right), \nonumber\\
 & =\int\frac{d\mathbf{x}d\mathbf{v}}{\left(2\pi\right)^{n}} \nonumber\\
 &e^{-\frac{1}{2}\left[\mathbf{v}^{\top},\mathbf{x}^{\top}\right]\left[\begin{array}{cc}
\mathbf{I}_{n} & 2\sqrt{\sum_{i=2}^{n}a_{i}^{2}}\beta\mathbf{I}_{n} \nonumber\\
2\sqrt{\sum_{i=2}^{n}a_{i}^{2}}\beta\mathbf{I}_{n} & \left(1+2\beta \sum_{i=2}^{n}a_{i}^{2} \right)\mathbf{I}_{n}
\end{array}\right]\left[\begin{array}{c}
\mathbf{v} \nonumber\\
\mathbf{x}
\end{array}\right]}\\
 & =\det\left(\left[\begin{array}{cc}
 \mathbf{I}_{n} & 2\sqrt{\sum_{i=2}^{n}a_{i}^{2}}\beta\mathbf{I}_{n} \nonumber\\
 2\sqrt{\sum_{i=2}^{n}a_{i}^{2}}\beta\mathbf{I}_{n} & \left(1+2\beta \sum_{i=2}^{n}a_{i}^{2} \right)\mathbf{I}_{n}
 \end{array}\right]\right)^{-1/2} \nonumber\\
 & =\left(\frac{1}{1+2\beta\left(1-2\beta\right)\sum_{i=2}^{n}a_{i}^{2}}\right)^{n/2}. \label{gaussEq1}
\end{align}
By optimizing over $\beta$ in the denominator,
 we have $1+2\beta\left(1-2\beta\right)\sum_{i=2}^{n}a_{i}^{2} \leq 
1 + \frac{1}{4} \sum_{i=2}^{n}a_{i}^{2} $.
This means we can use $\beta=\frac{1}{4}$ to reach the tightest bound for  inequality (\ref{gaussEq1}).
Therefore, for any $\varepsilon>0$, we have

\begin{align*}
 & \lim_{n\rightarrow\infty}\mathrm{Pr}\left(\left\Vert \mathbf{b}_{1}+\sum_{i=2}^{n}\mathbf{b}_{i}a_{i}\right\Vert ^{2}\leq\left\Vert \mathbf{b}_{1}\right\Vert ^{2}\right)\\
 & \leq\lim_{n\rightarrow\infty}\left(\frac{1}{1+2\beta\left(1-2\beta\right)\sum_{i=2}^{n}a_{i}^{2}}\right)^{n/2}\\
 & <\varepsilon.
\end{align*}
\end{IEEEproof}
Next, we investigate the reduction on the dual of a Gaussian random
basis, which arises in our detection problem. For an input basis $\mathbf{B}$
we define
\[
\theta_{i,j}=\arccos\left(\frac{\left|\langle\mathbf{b}_{i},\mathbf{b}_{j}\rangle\right|}{\left\Vert \mathbf{b}_{i}\right\Vert \left\Vert \mathbf{b}_{j}\right\Vert }\right),1\leq i\neq j\leq n.
\]

The following lemma says that the SR-Pair method can provide a Gauss-reduced
basis for all pairs of vectors with pairwise angles $\theta_{i,j}>\pi/3$.
\begin{lem}
\label{claim: PR properties} For an SR-Pair reduced basis, we have
$\theta_{i,j}>\pi/3$ for all $i\neq j$.
\end{lem}
\begin{proof} If a lattice basis $\mathbf{B}$ is non-reducible by
SR-Pair, we have $\lfloor\langle\mathbf{b}_{i},\mathbf{b}_{j}\rangle/\langle\mathbf{b}_{j},\mathbf{b}_{j}\rangle\rceil=0$
$\forall i\neq j$. Therefore the lemma follows from
\[
\cos\theta_{i,j}<\frac{1}{2}\frac{\left\Vert \mathbf{b}_{i}\right\Vert }{\left\Vert \mathbf{b}_{j}\right\Vert }<\frac{1}{2}\frac{\min(\left\Vert \mathbf{b}_{i}\right\Vert ,\left\Vert \mathbf{b}_{j}\right\Vert )}{\max(\left\Vert \mathbf{b}_{i}\right\Vert ,\left\Vert \mathbf{b}_{j}\right\Vert )}\leq1/2.
\]

\end{proof} 

Here we argue that the pairwise angles are dense in the dual of a
Gaussian random basis. Fig. \ref{fig:histoPrimeDual}-(a) plots the
histogram of such random matrices. It shows that so a large number
of vectors satisfy $\theta_{i,j}<\pi/3$, and these vectors will trigger
the reduction in SR-Pair/SR-Hash. On the contrary, as predicted by
Proposition \ref{prop:nonReduceSLR}, Fig. \ref{fig:histoPrimeDual}-(b)
shows that the primal basis will not be
reduced by SR-Pair/SR-Hash. \textcolor{black}{Bases with dense angles also feature large orthogonality defects. In Fig. \ref{fig:oddemo}, we plot the OD versus dimension $n$ relations respectively for the dual and primal Gaussian random matrices. The figure shows the dual bases approximately have a growth rate of $O(20^{n^{1.5}})$, while that of the primal basis is extremely small. The above confirms that the objective lattice bases in MIMO detection are easily reducible by tuning the pairwise angles.} 

\begin{figure}[htb]
\center

\subfloat[Dual bases.]{\includegraphics[width=0.25\textwidth,height=0.4\textwidth]{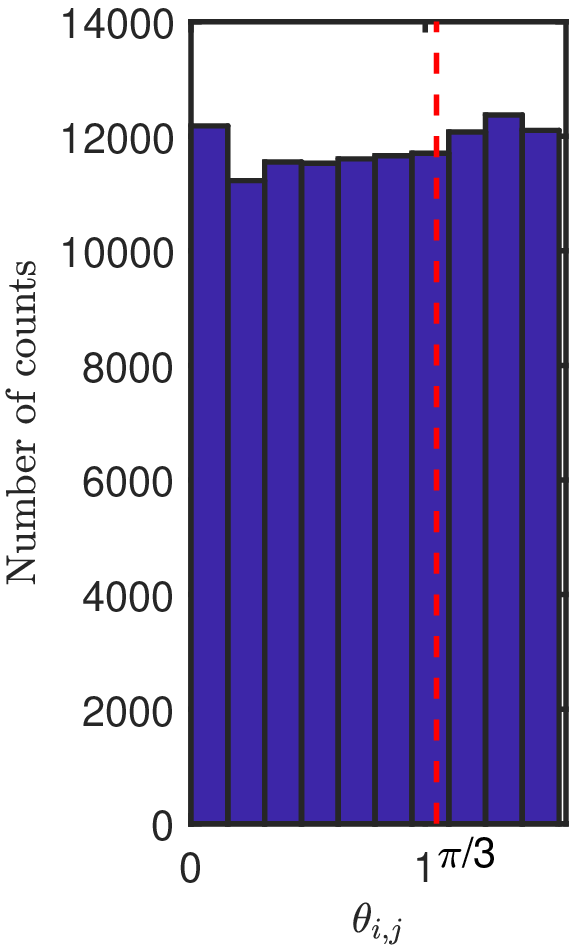}

}\subfloat[Primal bases.]{\includegraphics[width=0.25\textwidth,height=0.4\textwidth]{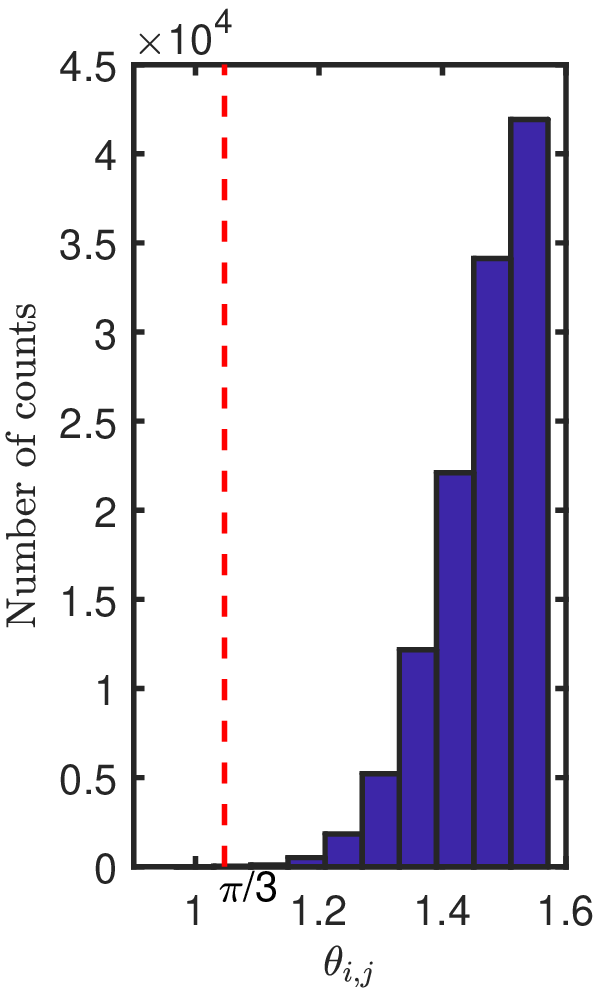}

}\caption{The histogram of the pairwise angles for the dual and primal of $60\times60$
Gaussian random bases.}
\label{fig:histoPrimeDual}
\end{figure}

\begin{figure}[htb]
	\center
	
	\includegraphics[width=0.4\textwidth]{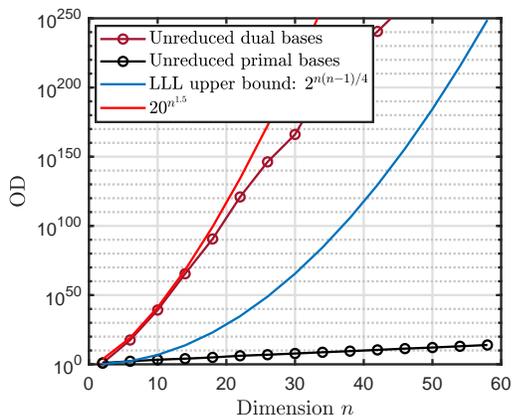}
	
	\caption{\textcolor{black}{The ODs of square Gaussian random bases.}}
	\label{fig:oddemo}
\end{figure}

\bibliographystyle{IEEEtranMine}
\bibliography{lib}

\begin{IEEEbiographynophoto}{Shanxiang Lyu}
	received the B.S. and M.S. degrees in electronic and information engineering from
	South China University of Technology, Guangzhou, China, in 2011
	and 2014, respectively, and the Ph.D. degree from the
	Electrical and Electronic Engineering Department, Imperial College London,
	in 2018. 
	He is currently a lecturer   
	with the College of Cyber Security, Jinan University.
	His
	research interests are in lattice theory, algebraic number theory, and their applications.
\end{IEEEbiographynophoto}

\begin{IEEEbiographynophoto}{Jinming Wen}
received his Bachelor degree in Information and Computing Science from Jilin Institute of Chemical Technology, Jilin, China, in 2008, his M.Sc. degree in Pure Mathematics from the Mathematics Institute of Jilin University, Jilin, China, in 2010, and his Ph.D degree in Applied Mathematics from McGill University, Montreal, Canada, in 2015. He was a postdoctoral research fellow at Laboratoire LIP (from March 2015 to August 2016), University of Alberta (from September 2016 to August 2017) and University of Toronto (from September 2017 to August 2018). He has been a full professor in Jinan University, Guangzhou since September 2018. His research interests are in the areas of lattice reduction and sparse recovery. He has published around 50 papers in top journals (including Applied and Computational Harmonic Analysis, IEEE Transactions on Information Theory/Signal Processing/Wireless Communications) and conferences. He is an Associate Editor of IEEE Access.
\end{IEEEbiographynophoto}

\begin{IEEEbiographynophoto}{Jian Weng}
received the B.S. and M.S. degrees
in computer science from South China University of Technology, Guangzhou, China, in 2001
and 2004, respectively, and the Ph.D. degree in
computer science from Shanghai Jiao Tong University, Shanghai, China, in 2008.
He is currently a Professor and the Executive
Dean with the College of Information Science
and Technology, Jinan University, Guangzhou,
China. He has authored/coauthored 80 papers
in international conferences and journals, such
as CRYPTO, EUROCRYPT, ASIACRYPT, TCC, PKC, CT-RSA, IEEE
TPAMI, IEEE TDSC, etc. His research areas include public key cryptography, cloud security, blockchain, etc.
Prof. Weng was a recipient of the Young Scientists Fund of the
National Natural Science Foundation of China in 2018, and the Cryptography Innovation Award from Chinese Association for Cryptologic
Research (CACR) in 2015. He served as the General Co-Chair for SecureComm 2016, TPC Co-Chairs for RFIDsec13 Asia, and ISPEC 2011,
and program committee members for more than 40 international cryptography and information security conferences. He also serves as an
Associate Editor for IEEE TRANSACTIONS ON VEHICULAR TECHNOLOGY.
\end{IEEEbiographynophoto}

\begin{IEEEbiographynophoto}{Cong Ling} (S'99-A'01-M'04) 
	received the B.S. and M.S. degrees in electrical engineering from
	the Nanjing Institute of Communications Engineering, Nanjing, China, in 1995
	and 1997, respectively, and the Ph.D. degree in electrical engineering from
	the Nanyang Technological University, Singapore, in 2005.
	He had been on the faculties of the Nanjing Institute of Communications
	Engineering and King's College. He is currently a Reader (Associate Professor) with the Electrical and Electronic Engineering Department, Imperial
	College London. His research interests are coding, information theory, and
	security, with a focus on lattices.
	Dr. Ling has served as an Associate Editor for the IEEE TRANSACTIONS
	ON COMMUNICATIONS and the IEEE TRANSACTIONS ON VEHICULAR
	TECHNOLOGY.\end{IEEEbiographynophoto}

\end{document}